\documentclass[iicol,sn-mathphys-num]{sn-jnl}
\usepackage{graphicx}%
\usepackage{multirow}%
\usepackage{amsmath,amssymb,amsfonts}%
\usepackage{amsthm}%
\usepackage{mathrsfs}%
\usepackage[title]{appendix}%
\usepackage{xcolor}%
\usepackage{textcomp}%
\usepackage{manyfoot}%
\usepackage{booktabs}%
\usepackage{algorithm}%
\usepackage{algorithmicx}%
\usepackage{algpseudocode}%
\usepackage{listings}%

\theoremstyle{thmstyleone}%
\newtheorem{theorem}{Theorem}%  meant for continuous numbers
\newtheorem{proposition}[theorem]{Proposition}% 

\theoremstyle{thmstyletwo}%
\newtheorem{remark}{Remark}%

\theoremstyle{thmstylethree}%
\newtheorem{definition}{Definition}%

\usepackage{amsmath,amssymb,amsfonts,mathrsfs}
\usepackage{mathtools}
\usepackage[utf8]{inputenc}
\usepackage[normalem]{ulem}
\usepackage{subcaption}

\graphicspath{{figures/}}
\allowdisplaybreaks

\def \R {{\mathbb {R}}}
\def \C {{\mathbb {C}}}
\def \Q {{\mathbb {Q}}}
\def \Z {{\mathbb {Z}}}
\def \N {{\mathbb {N}}}

\DeclarePairedDelimiter{\abs}{\lvert}{\rvert}
\DeclarePairedDelimiter{\norm}{\lVert}{\rVert}

\DeclareMathOperator{\sgn}{sgn}
\DeclareMathOperator{\Res}{Res}
\DeclareMathOperator{\Id}{\mathbb{I}}

\DeclareMathOperator*{\esssup}{ess\,sup}

\newtheorem{lemma}[theorem]{Lemma}
\newtheorem{corollary}[theorem]{Corollary}

\begin{document}

\title{Neural Field Equations with Time-Periodic External Inputs and Some Applications to Visual Processing}

\author*[1]{\fnm{M.~Virginia} \sur{Bolelli}}\email{maria-virginia.bolelli@centralesupelec.fr}
\author[1,2]{\fnm{Dario} \sur{Prandi}}\email{dario.prandi@centralesupelec.fr}
\equalcont{These authors contributed equally to this work.}

\affil*[1]{\orgdiv{Laboratoire des Signaux et Syst\'emes} \orgname{Universit\'e Paris-Saclay, CentraleSup\'elec}
	\orgaddress{\street{3 rue Joliot Curie}, \city{Gif-sur-Yvette}, \postcode{91190}, \country{France}}
}

\affil[2]{\orgname{CNRS}}

\abstract{
The aim of this work is to present a mathematical framework for the study of flickering inputs in visual processing tasks. When combined with geometric patterns, these inputs influence and induce interesting psychophysical phenomena, such as the MacKay and the Billock-Tsou effects, where the subjects perceive specific afterimages typically modulated by the flickering frequency. 
Due to the symmetry-breaking structure of the inputs, classical bifurcation theory and multi-scale analysis techniques are not very effective in our context. We thus take an approach based on the input-output framework of control theory for Amari-type neural fields. This allows us to prove that, when driven by periodic inputs, the dynamics converge to a periodic state. Moreover, we study under which assumptions these nonlinear dynamics can be effectively linearised, and in this case we present a precise approximation of the integral kernel for short-range excitatory and long-range inhibitory neuronal interactions. 
Finally, for inputs concentrated at the center of the visual field with a flickering background, we directly relate the width of the illusory contours appearing in the afterimage with both the flickering frequency and the strength of the inhibition.
 }


\maketitle

\section{Introduction}
In computational neuroscience, neural fields are fundamental models for understanding the collective behavior of neuronal populations. Initially developed in the 70s to model the sensory visual cortex \cite{A72, WC72, ermentrout1979mathematical}, they have been refined and adapted to account for various factors such as the functional architecture of the primary visual cortex \cite{bressloff2003functional,bressloff2001geometric, veltz2015effects, faye2011analysis, SC15, bolelli2024individuation, BCSZ23b}, and temporal delays \cite{FCF10, faye2010some}.
These models are highly suitable for describing phenomena like spontaneous geometric visual hallucinations that emerge in the visual field due to sudden qualitative changes in specific physiological parameters \cite{ermentrout1979mathematical,bressloff2001geometric}. They also prove effective in understanding sensory-driven and self-organized cortical activity interactions when the visual stimulus exhibits regular shape and complete distribution across the visual field, with symmetry respecting a subgroup of the Euclidean group \cite{nicks2021understanding}.

In this study, we use neural fields to understand the impact of flickering external inputs on the perception of visual phenomena. 
This situation has been investigated in \cite{rule2011model}, in a slightly different setting, for spatially homogeneous inputs.
In this work, continuing the research line started in \cite{tamekue2023mathematical,tamekue2024reproducibility}, we are interested in visual stimuli that do not necessary exhibit symmetries. 
As such, the standard tools from equivariant bifurcation theory cannot be employed and we resort to an approach inspired by control theory, where the external input is seen as the control.

Specifically, we consider the following Amari-type neural field equation,
\begin{equation}
\label{eq:mean_field}
\frac{\partial u}{\partial t}  = -u + \omega \ast f(u) + I.
\end{equation}
This equation is characterized by a spatial interaction kernel $\omega$, a firing rate function $f$, and a time-periodic external input $I$, where $\ast$ denotes convolution in space.  Here, $I$ captures the cortical response to flickering visual stimuli, and the solution $u$ represents the perceived visual pattern.

Since the aim of this paper is to understand the effect of flickering visual stimuli under normal conditions, the parameters defining the mean field equation, such as those for the kernel $\omega$ or the firing rate function $f$, are chosen so that, in absence of input (e.g., $I \equiv 0$), a unique stable stationary solution for \eqref{eq:mean_field} exists. 

Under these assumptions, we theoretically study equation \eqref{eq:mean_field} allowing to predict cortical dynamics under a time-periodic inputs $I$. 
In particular, we show that introducing a time-periodic input $I$ ensures the dynamic to converge toward a globally attractive periodic state. Additionally, the solution of the neural field equation respects the symmetries of the vector field (r.h.s. of \eqref{eq:mean_field}), entailing that a spatially homogeneous input will produce a spatially homogeneous solution. This differs from the model proposed in \cite{rule2011model} for studying flicker-induced phosphenes, where the input affects system stability, and a homogeneous input results in an non-homogeneous solution.

We then show under which conditions it is possible to approximate the nonlinear dynamics \eqref{eq:mean_field} with their linearisation, and study the latter via Fourier techniques. Namely, when $\omega$ is taken to be a difference of Gaussians, we explicitly calculate the integral kernels solving the linear version of \eqref{eq:mean_field}, thus generalizing a technique introduced for the static case in \cite{tamekue2023cortical, tamekue2023mathematical}. This allows not only to exhibit an explicit expression for the attractive periodic state  $u^\star$ in the one-dimensional case, but also in the general case, albeit only for \emph{one-dimensional inputs} (i.e., inputs that depend on a single variable). 

These theoretical results are particularly relevant to study flicker-induced visual phenomena.
Indeed, geometric patterns commonly studied in visual processing, and particularly visual illusions, often include concentric circles or radial shapes whose induced cortical input is effectively represented by one-dimensional inputs, due to the retinotopic mapping from the retina to V1.

Our research inquires on how localised flicker inputs interact with these geometric patterns: leveraging analysis from the linear case, we explicitly derive solutions corresponding to these time-dependent inputs. In accordance with \cite{tamekue2023cortical, tamekue2023mathematical}, localized information generates geometric patterns complementary to its position, and we show how flicker induces their perceived motion. We also explore how kernel and input parameters influence solution dynamics, affecting the perceived pattern size. 
Indeed, we theoretically show that the thickness of the illusory contours induced by an input localised at the center of the fovea depend on the frequency of the background flickering.
Moreover, we show how this phenomenon is strictly connected with the relative strength of the long-range inhibition.
Finally, we apply these findings to psychophysical experiments like the MacKay effect \cite{mackay1957moving} and Billock-Tsou illusion \cite{billock2007neural}, showcasing how parameter variations impact the size and reproducibility of perceived phenomena.

\subsection{Structure of the paper}
The paper is organized as follows.
In {Section~\ref{sec:general}}, we present the neural field model with periodic input. We address the well-posedness of the Cauchy problem, the existence of a globally attractive periodic solution, and we study its behavior in the presence of quasi-linear firing rate functions. We show that if the input is sufficiently small, the solution to the general mean field equation behaves like the solution to the equation characterized by the linear firing rate function. 

{Section~\ref{sec:linear}} deals with the linear mean field equation. We decompose the globally attractive periodic solution using Fourier transform in space and Fourier expansion in time for the input, introducing the convolution with a suitable kernel. We restrict ourselves to the 1D case, where we can directly compute these kernels using results from complex analysis, extending a technique proposed in \cite{tamekue2023mathematical}. The proof of some necessary technical results are postponed to Appendix~\ref{app:proofs}.

Finally, in {Section~\ref{sec:application}}, we use the explicit representation of the globally attractive state of the mean field equation to apply the theoretical results to the domain of visual processing. Firstly, we investigate the solution in response to a localized flickering visual stimulus in the linear case of the firing rate function. Then, we examine how the parameters defining the equation, including kernel and input parameters, influence the solutions. We explore how these findings can be applied to psychophysical experiments that can be approached from a one-dimensional perspective, such as the MacKay effect and the Billock-Tsou illusion.

\subsection*{Notation}

The Euclidean norm of a point $x \in \R^d$ is denoted by $\abs{x}$, and $\langle \cdot, \cdot \rangle$ denotes the scalar product in $\R^d$. 
Given a complex number $z\in \C$, we denote with $\Re z$ and $\Im z$ its real and imaginary part, respectively. 

For $1 \leq p \leq \infty$, we consider the Lebesgue functional space $L^p(\R^d)$ with its standard norm $\norm{\cdot}_p$. We define $L^\infty([0, \infty); L^p(\R^d))$ as the space of all real-valued functions $u$ on $\R^d \times [0, \infty)$ such that 
\[
\norm{u}_{L_t^\infty L_x^p} \coloneqq \esssup_{t \geq 0}\, \norm{u(\cdot, t)}_{p} < \infty.
\]
We endow this space with the norm $\norm{u}_{L_t^\infty L_x^p}$. Note that for $u \in L^\infty([0, \infty); L^p(\R^d))$ and almost every $t\in [0,+\infty)$ we have that $u(\cdot, t) \in L^p(\R^d)$.

We let $\hat{u}$ be the Fourier transform of a function $u\in L^p(\mathbb{R}^d)$. When $u\in L^p(\mathbb{R}^d)\cap L^1(\mathbb{R}^d)$, this is given by 
\begin{equation}
   \hat{u}(\xi) = \int_{\R^d} u(x) e^{-2\pi i \langle x, \xi \rangle} \, dx.
\end{equation}
The inversion formula then reads
\[
   u(x) = \int_{\R^d} \hat{u}(\xi) e^{2\pi i \langle x, \xi \rangle} \, dx.
\]

\section{Neural field model with periodic input}
\label{sec:general}

In this section we collect results concerning the neural field equation \eqref{eq:mean_field}, 
for general, and in particular nonlinear, activation functions.

We consider the following assumptions on the parameters:
\begin{enumerate}
	\item The spatial interaction kernel $\omega:\R^d\to\R$ is an even function (i.e., $\omega(x)=\omega(-x)$), and it satisfies $\|\omega\|_1<1$.
	\item The firing rate function $f:\R\to\R$ is Lipschitz continuous, such that $f(0)=0$, $\max_{s \in \R} f'(s)=f'(0)=1$, and $\|f'\|_\infty \le 1$. 
	\item The external visual input $I$ belongs to $L^\infty([0,+\infty),L^p(\R^d))$.
\end{enumerate}

\begin{remark}
	In the static case, namely when the input $I(x,t) = I(x)$ is time-independent, these assumptions ensure that there is a unique stable stationary state associated with the mean field equation \eqref{eq:mean_field}, see \cite{tamekue2023mathematical, tamekue2024reproducibility}.
\end{remark}

\subsection{Well-posedness of the Cauchy problem}
We use classical tools of functional analysis in Banach spaces (for a general reference, see \cite{deimling2006ordinary,brezis2010functional}), to show the property of existence and uniqueness of a solution to equation \eqref{eq:mean_field}.
Assume $1\leq p \leq \infty$, 
we define the operator $A: L^p(\R^d) \longrightarrow L^p(\R^d)$ as:
\begin{equation}
\label{eq:operator_notation}
    A \coloneqq -\Id + W_f
\end{equation}
with $\Id $ identity operator and
\begin{equation}
W_fu(x) = [\omega\ast f(u)](x) =\int_{\R^d} \omega(x-y)f(u(y))dy.
\end{equation}
We recast equation \eqref{eq:mean_field} as a Cauchy problem in $L^p(\R^d)$:
\begin{equation}
\label{eq:cp_mean_field}
\begin{cases}
    u'(t) = F(t, u)\\
    u(0)  = u_0 , \hspace{0.5cm} u_0 \in L^p(\R^d)
\end{cases}
\end{equation} 
with $F(t, u) = Au(t) + I(t)$ and  $I \in L^\infty ([0, \infty); L^p(\R^d))$.

\begin{definition}
A solution of the Cauchy problem \eqref{eq:cp_mean_field} is an absolutely continuous function $u : [0, \infty) \to L^p(\R^d)$, which is differentiable almost everywhere on $[0, \infty)$ and satisfies equation \eqref{eq:cp_mean_field} for almost every $t \in [0, \infty)$ with initial condition $u(0) = u_0$.
\end{definition}

% \begin{remark}
%     If $p\in (1,\infty)$, then an absolutely continuous function is automatically differentiable almost everywhere. If $p=1$ or $p=\infty$, this is not always the case, as $L^p$ is not reflexive.
% \end{remark}

\begin{proposition}
\label{prop:well-pos_(CP)}
For any initial condition $u_0 \in L^p(\R^d)$ the Cauchy Problem \eqref{eq:cp_mean_field} has a unique solution taking values in $L^p(\R^d)$, that is locally Lipschitz continuous on $[0,\infty)$. 
\end{proposition}

\begin{proof}
Firstly, we show that $F(t,\cdot)$ is Lipschitz in $L^p(\R^d)$. Suppose $u_1, u_2 \in L^p(\R^d)$,
by Young's convolutional inequality, we have
\begin{equation}
\begin{aligned}
\norm{\omega\ast [f(u_1)-f(u_2)]}_p^p  \leq
 \norm{\omega}_1^{p}\norm{f(u_1)-f(u_2)}_p^p. 
\end{aligned}
\end{equation}
This yields
\begin{equation}
    \begin{aligned}
        \norm{F(t, u_1)-F(t, u_2)}_p  \leq 
        (1+\norm{\omega}_1)\norm{u_1-u_2}_p.
    \end{aligned}
\end{equation}
The existence of a Lipschitz solution to the Cauchy Problem \eqref{eq:cp_mean_field} follows by \cite[Thm. 8.1]{deimling2006ordinary} (see also \cite[Thm. 7.3]{brezis2010functional}). Such solution is unique as a consequence of Gronwall's Lemma, thanks to the Lipschitz continuity of $F$. 
\end{proof}

In order to study symmetries of solutions to \eqref{eq:cp_mean_field}, we introduce the following.

\begin{definition}
Given a map $\phi:\mathbb{R}^d\to \mathbb{R}^d$, we say that a function $u:\mathbb{R}^d\to \mathbb{R}$ is $\phi$-invariant if $u\circ \phi = u$. Moreover, an operator $A$ is $\phi$-equivariant if $A(u\circ \phi)=A(u) \circ \phi$ for any $u$ in its domain.
\end{definition}

The following result shows that the solution to \eqref{eq:cp_mean_field} respects the symmetries of the equation. 
%A function $f : \mathbb{R}^d \to \mathbb{R}$ is said to be \textit{$\phi$-invariant} with respect to a transformation $\phi : \mathbb{R}^d \to \mathbb{R}^d$ if it satisfies $f \circ \phi = f$. 
%That is, for every $x \in \mathbb{R}^d$, $f(\phi(x)) = f(x)$.

\begin{proposition}
Consider $\phi\in C^1(\R^d;\R^d)$ 
with Jacobian  $J_{\phi}$ s.t. $\abs{\det(J_{\phi})}=1$. 
Assume that the operator $W_f$ is $\phi$-equivariant. Then, for any $\phi$-invariant input and any $\phi$-invariant initial condition, the solution to the Cauchy Problem \eqref{eq:cp_mean_field} is $\phi$-invariant.
\end{proposition}

\begin{proof}
 Consider $u(x, t)$ solution of \eqref{eq:cp_mean_field} and let $v(x, t) = u(\phi(x),t)$. Then, $v$ satisfies:
\begin{equation}
\begin{cases}
    \partial_t v(x,t) = -v(x,t) + \omega\ast f(v(\cdot, t))(x) + I(x,t)\\
    v(x,0)=u_0(x).
    \end{cases}
\end{equation}
The result follows from the uniqueness of the solution. 
\end{proof}

Since $W_f$ is a convolution operator, it is $\phi$-equivariant for all translations $\phi(x)=x-a$, $a\in\R$. Hence, spatially homogeneous inputs and initial conditions yield spatially homogeneous solutions. 

In general, $W_f$ is $\phi$-equivariant if $\omega(x-\phi^{-1}(z))=\omega(\phi(x)-z)$ for all $x,z\in\R$, and hence additional symmetries depend on the kernel $\omega$. We mention that if $\omega$ is radial (i.e., $\omega(x)=\omega(|x|)$) then $W_f$ is equivariant w.r.t.~rotations.

\subsection{Existence of periodic solution}
In this section we study the long-term behavior of the dynamical system governed by  \eqref{eq:mean_field}, when the external input is periodic.

\begin{theorem}
\label{thm:attractive_limit_cycle} 
If the input $I \in L^\infty ([0, \infty); L^p(\R^d))$ is $T$- periodic in time, then \eqref{eq:mean_field} has a unique $T$-periodic solution $u^\star(t) \in L^p(\R^d)$ such that, for any other solution $u$ of \eqref{eq:mean_field} it holds
\begin{equation}
\| u(t) - u^\star(t) \|_p \le e^{-t(1-\|\omega\|_1)}\|u(0)-u^\star(0)\|_p.
\end{equation}

\end{theorem}

\begin{proof}
Consider the Poincar\'e map $\Phi$ associated to the mean field equation \eqref{eq:mean_field}:
\begin{equation}
\label{eq:PM}
\begin{aligned}
    \Phi: \:& L^p(\R^d) && \longrightarrow && L^p(\R^d) \\
          & u_0 && \longmapsto && u_1
\end{aligned}
\end{equation}
where $u_1: = u(T)$ with $u(t)$ denoting the trajectory associated to the differential equation \eqref{eq:mean_field} with initial condition $u(0)=u_0$.
The idea is to prove that the Poincar\'e map $\Phi$ is a contraction. As such, it will have a unique fixed point. Due to the periodicity of the input and the uniqueness
of the solution, this yields the existence and the uniqueness of the periodic solution. The attractiveness follows from standard estimates. 

Let us consider  $u(t)$ solution for the Cauchy Problem \eqref{eq:cp_mean_field} associated to \eqref{eq:mean_field} with initial condition $u_0$ and let $v(t)$ be the solution to \eqref{eq:cp_mean_field} with initial condition $v_0$. 
    Then, by Variation of Constant formula, and taking the $L^p$ norm, we have: 
\begin{equation}
\begin{split}
    e^t&\|u(t)-v(t)\|_p-\norm{u_0-v_0}_p\\
    &\leq    \int_0^t e^s\norm{\omega \ast [f(u(s))-f(v(s))]}_pds\\
     &\le \norm{\omega}_1 \int_0^te^s \norm{u(s)-v(s)}_pds.\\
\end{split}
\end{equation}
Rewriting the inequality as: 
\begin{equation}
    z(t)\coloneqq e^t \norm{u(t)-v(t)}_p
     \leq   z(0) + \norm{\omega}_1 \int_0^tz(s)ds,
\end{equation}
we can apply Gronwall's Lemma (see Lemma~\ref{lem::gronwall-lemma}) to get: 
\begin{equation}
    z(t)
     \leq   z(0) e^{\int_0^t\norm{\omega}_1 ds} = z(0)e^{t\norm{\omega}_1}.
\end{equation}
It follows then, that for every $t$ we have:
\begin{equation}
\label{eq:dis-g}
    \norm{u(t)-v(t)}_p
    \leq  e^{-t(1-\norm{\omega}_1)} \norm{u_0-v_0}_p.
\end{equation}
Finally, if we use \eqref{eq:dis-g} for $t = T$, we obtain: 
\begin{equation}
    \begin{aligned}
       \norm{\Phi(u_0)-\Phi(v_0)}_p \leq  
         e^{-T(1-\norm{\omega}_1)}\norm{u_0-v_0}_p.\\
    \end{aligned}
\end{equation}
This implies the Poincar\'e map $\Phi$ to be contractive since the period $T$ is positive and $1-\norm{\omega}_1>0$. 

Let $u^\star$ be the unique periodic solution of  \eqref{eq:mean_field}; the thesis follows by applying \eqref{eq:dis-g} to $u^\star$ and to any other solution $u$ of \eqref{eq:mean_field}. 
\end{proof}

It is worth noting that if the firing rate function $f$ is linear, we can identify the initial condition that yields the periodic solution.

\begin{proposition}
\label{rmk:initial-condition-periodic}
    Let $1\leq p \leq \infty$ and assume $f(s)=s$. If the input $I \in L^\infty ([0, \infty); L^p(\R^d))$ is $T$-periodic in time, the periodic solution $u^\star$ is 
	the solution of \eqref{eq:mean_field} passing through the point
    \begin{equation}
\label{eq:init_cond}
    u_{0, I} \coloneqq (\Id-e^{TA})^{-1}\left[ -\int_{0}^T e^{-sA}I(s)ds\right].
\end{equation}
\end{proposition}
\begin{proof}
Let us use the operator notation introduced in \eqref{eq:operator_notation}. The Variation of Constants formula allows us to write, for a given $u_0 \in L^p(\R^d)$, the unique solution $u(t)$ to the Cauchy Problem associated to equation \eqref{eq:mean_field-lin} as:
\begin{equation}
\label{eq:variation_constant}
    u(t) = e^{tA}u_0 + e^{tA}\int_0^t e^{-sA}I(s)ds. 
\end{equation}
Taking as initial condition the function\footnote{$(\Id-e^{TA})$ is invertible since $A$ is a contraction thanks to $\norm{\omega}_1 < 1$. }
\begin{equation}
    u_{0, I} \coloneqq (\Id-e^{TA})^{-1}\left[ -\int_{0}^T e^{-sA}I(s)ds\right]
\end{equation}
independent of time but dependent on the non-homogeneous external input $I$. Then, the solution $u(t)$ for the Cauchy Problem \eqref{eq:cp_mean_field} is periodic with the same periodicity as $I$:     
\begin{multline}
                 u(t) - u(t+T)  \\
                 =e^{tA}{\left[ (\Id-e^{TA})u_{0, I}+\int_0^Te^{-sA}I(s)ds \right]}
                 =  0.
\end{multline}
This completes the proof of the statement.
\end{proof}

\subsection{Quasi-linear nonlinearities}

We investigate under which circumstances having a linear or nearly linear firing rate function provides sufficient information to understand the problem. Specifically, we observe when the presence of a nonlinear firing rate function $f$ does not produce a significantly different outcome compared to the scenario where the equation is linear.

\subsubsection{Prior estimates}

\begin{proposition}
\label{prop:upper-bound-sol}
Let $1\leq p \leq \infty$. For any solution $u$ of \eqref{eq:mean_field} we have
    \begin{equation}
        \|u(t)\|_p \le e^{-t(1-\|\omega\|_1)}\|u(0)\|_p + \frac{\norm{I}_{L_x^p \times L_t^\infty}}{1-\norm{\omega}_1}.
    \end{equation}
\end{proposition}
\begin{proof}
Writing the solution to \eqref{eq:mean_field} using the Variation of Constant formula with initial condition $u_0 \in L^p(\R^d)$ and taking its $L^p$ norm, we get: 
    \begin{equation}
    \label{eq:G1}
    \begin{aligned}
        e^t&\norm{u(\cdot, t)}_{p} - \norm{u_0}_{p}    \\
         &\le \int_0^t e^{s} \norm{\omega}_1\norm{u(\cdot, s)}_{p}ds+  \int_0^t e^{s} \norm{I(\cdot, s)}_{p}ds\\
         &\le\norm{\omega}_1 \int_0^t e^{s} \norm{u(\cdot, s)}_{p}ds+  \norm{I}_{L_x^p L_t^\infty}\int_0^t e^s\,ds.\\
        \end{aligned}
    \end{equation}
Rewriting the inequality as: 
\begin{multline}
\label{eq:G2}
    v(t)\coloneqq e^t \norm{u(\cdot, t)}_{p} \leq \\ v(0) +  \norm{\omega}_1 \int_0^tv(s) ds + \norm{I}_{L_x^p L_t^\infty} \int_0^t e^s\,ds,
\end{multline}
using Gronwall's Lemma (see Lemma~\ref{lem::gronwall-lemma}), we have: 
\begin{equation}
\label{eq:G3}
\begin{aligned}
    v(t) 
&\leq  v(0)e^{\norm{\omega}_1t}+ \frac{\norm{I}_{L_x^p L_t^\infty} }{1-\norm{\omega}_1}(e^t-e^{\norm{\omega}_1 t}).\\
\end{aligned}  
\end{equation}
The statement immediately follows.
\end{proof}

If the firing rate function $f$ is  linear, the estimate can be improved. 

\begin{proposition}
\label{rmk:linear-upper-bound}
Let $1\leq p \leq \infty$ and assume that $f(s)=s$. Then, for any solution $u$ of \eqref{eq:mean_field} we have
    \begin{equation}
    \limsup_{t \rightarrow \infty} \norm{u(\cdot, t)}_{p} \leq \norm{I}_{L_x^p \times L_t^\infty}.
    \end{equation}
\end{proposition}

\begin{proof}
Let $A$ be the $L^p$-operator defined in \eqref{eq:operator_notation}.
For any initial condition $u_0 \in L^p(\R^d)$, we can take the $L^p$ norm of the solution written using the Variation of Constant formula \eqref{eq:variation_constant}, to obtain:
    \begin{equation}
        \norm{u(\cdot, t)}_{p} 
        \leq e^{-t\norm{A}} \norm{u_0}_{p} + (1-e^{-t\norm{A}})\norm{I}_{L_x^p\times L_t^\infty}
    \end{equation}
    The statement follows by taking the limit as $t$ tends to $+\infty$, since $\norm{A}<1$ due to the fact that $\norm{\omega}_1<1$.
\end{proof}

\subsubsection{Linear approximation}

In this section we compare the solutions to \eqref{eq:mean_field} with a nonlinear and a linear firing rate function $f:\R\to\R$, for small inputs. More precisely, for a given input $I\in L^\infty([0,+\infty),L^p(\R^d))$ and for any $\varepsilon>0$, we let $u_{\textsc{NL}}(\varepsilon)$ be the solution to \eqref{eq:mean_field} with input $\varepsilon I$, and initial condition $u_{\textsc{NL}}(\varepsilon)|_{t=0}\equiv 0$. The solution to the linear equation 
\begin{equation}
\label{eq:mean_field-lin-eps}
\frac{\partial u}{\partial t}= - u + \omega\ast u + \varepsilon I, \qquad u(0)=0,
\end{equation}
is denoted by $\varepsilon u_{\textsc{L}}$, where $u_{\textsc{L}}$ is the solution to \eqref{eq:mean_field-lin-eps} with $\varepsilon=1$. Observe that this is true due to the choice of the initial condition.

\begin{proposition}
\label{prop:}
Let $1\leq p \leq \infty$, and assume that $f\in C^2(\R;\R)$ with $f''\in L^\infty(\R)$.
Then, there exists a constant $c>0$,  depending only on $\omega$, $I$, and $\|f''\|_\infty$, such that
\begin{equation}
\|u_{\textsc{NL}}(\varepsilon)-\varepsilon  u_{\textsc{L}}\|_p \le 
    c\varepsilon^2,
\end{equation}
In particular, if the input $I$ is $T$-periodic in time, letting $u^\star_{\textsc{NL}}(\varepsilon)$ and (resp.~$\varepsilon u^\star_{\textsc{L}}$) be the unique periodic solutions to the nonlinear (resp.~linear) neural field equation with input $\varepsilon I$, we have
\begin{equation}
    u^\star_{\textsc{NL}}(\varepsilon) = \varepsilon u^\star_{\textsc{L}} + O(\varepsilon^2).
\end{equation}
\end{proposition}

\begin{proof}
To lighten the notation, we drop the explicit dependence on $\varepsilon$.
Let us start by defining 
\begin{equation}
v = u_{\textsc{NL}} - \varepsilon  u_{\textsc{L}}.
\end{equation}
Then, $v(0)=0$ and $v$ satisfies
\begin{equation}
\label{eq:ode_v}
    \partial_t v = -v + \omega \ast [f(u_{\textsc{NL}})-\varepsilon u_{\textsc{L}}].
\end{equation}
In particular, by the Variation of Constant formula, the solution of \eqref{eq:ode_v} reads 
\begin{equation}
\label{eq:v-sol}
    v(t) = e^{-t} \int_0^t e^s [\omega \ast (f(u_{\textsc{NL}})-\varepsilon u_{\textsc{L}})](s) ds.
\end{equation}
 
The fact that $f''\in L^\infty(\R)$ allows to consider, for any $x\in \R^d$ and $t\ge 0$, the following
\begin{equation}
    \label{eq:taylor-dev}
        f(u_{\textsc{NL}}(x,t)) 
%        &=  f(0)+ f'(0)u_{\textsc{NL}}(x,t) + \frac{f''(\xi_x)u_{\textsc{NL}}(x,t)^2}{2}\\
        \leq  u_{\textsc{NL}}(x,t) + \frac{\norm{f''}_\infty u_{\textsc{NL}}(x,t)^2}{2}.
\end{equation}
Here, we used the Taylor development of $f(u_{\textsc{NL}}(x,t))$ in $0$, together with
%$\xi_x\in[-|u_{\textsc{NL}}(x,t)|,|u_{\textsc{NL}}(x,t)|]$ is given by the Taylor-Lagrange remainder expression, and
the fact that $f(0)=0$ and $f'(0)= 1$.

We now claim that there exists $c_1>0$, depending only on $\omega$ and $I$, such that $\|u_{\textsc{NL}}^2\|_p\le c_1\varepsilon^2$. Indeed, $u_{\textsc{NL}}^2$ solves the equation
\begin{equation}
    \partial_t w = - w + h, \qquad h = u_{\textsc{NL}}  \left(\omega\ast f(u_{\textsc{NL}}) + \varepsilon I\right),
\end{equation}
with initial condition $w(0)=0$. As a consequence, 
\begin{equation}
    \label{eq:u-squared}
    e^t\|u_{\textsc{NL}}^2\|_p \le \int_0^t e^s \|h(s)\|_p\,ds.
\end{equation}
By Proposition~\ref{prop:upper-bound-sol}, there exists $c_2>0$, depending only on $\omega$ and $I$, such that $\|u_{\textsc{NL}}\|_p \le c_2\varepsilon$, and hence at any fixed time $s>0$ we have
\begin{equation}
    \begin{split}
        \|h\|_p 
        &\le \|\omega\ast f(u_{\textsc{NL}})\|_\infty\|u_{\textsc{NL}}\|_p + \varepsilon \|u_{\textsc{NL}}\|_p\|I\|_\infty\\
        &\le \|\omega\|_q\|u_{\textsc{NL}}\|_p^2 + \varepsilon \|u_{\textsc{NL}}\|_p\|I\|_\infty\\    
        &\le c_2(c_2\|\omega\|_q +\|I\|_\infty)\varepsilon^2.
    \end{split}
\end{equation}
Here, in the second inequality, we used Young's convolution inequality together with the fact that $\|f'\|_\infty\le 1$.
Plugging the above in \eqref{eq:u-squared} proves the claim.

By the previous claim, \eqref{eq:v-sol}, and \eqref{eq:taylor-dev}, there exists a constant $c_3>0$, depending only on $\omega$, $I$, and $\|f''\|_\infty$, such that 
\begin{equation}
    \begin{aligned}
        e^t\|&v(t)\|_{p}  \\
        &\leq  \int_0^t e^s \norm{\omega}_1\left( \norm{v(s)}_{p}+ \frac{\norm{f''}_\infty \norm{u_{\textsc{NL}}^2}_p}{2}\right) ds  \\ 
                 &\leq  \norm{\omega}_1 \int_0^t e^s \norm{v(s)}_{p}ds  +  
         c_3\varepsilon^2 \int_0^t e^s\,ds.  
    \end{aligned}
\end{equation} 
Finally, the statement follows by applying Gronwall's Lemma (see Lemma~\ref{lem::gronwall-lemma}).
\end{proof}

The previous result shows that the solutions to the linear and to the nonlinear equation are similar in norm, when the input is sufficiently small. 
It is worth noticing that a similar result can be recovered by replacing the smallness assumption on the input by an assumption on the firing rate function.

\begin{proposition}
\label{prop:quasi-linearity}
    Assume that there exists $\eta>0$ such that $f(s)= s$ for any $s$ such that $\abs{s} \leq \norm{I}^2_{L^\infty_t L^p_x}+\eta$. Then, letting $u_{\textsc{NL}}$ be any solution to the nonlinear equation \eqref{eq:mean_field} and $u_{\textsc{L}}$ be a solution to the corresponding linear equation, we have
    \begin{equation}
       \norm{u_{\textsc{NL}}-u_\textsc{L}}_{p} \xrightarrow{t \rightarrow \infty }{0}.
    \end{equation}
\end{proposition}

\begin{proof}
We can write the solution $u_{\textsc{NL}}$ of the nonlinear equation \eqref{eq:mean_field} as sum of the linear  $u_{\textsc{L}}$ solution for the linear equation \eqref{eq:mean_field-lin-eps} plus a function $u$ in the following way:
    \begin{equation}
        u_{\textsc{NL}} = u + u_{\textsc{L}};
    \end{equation}
where $u$ satisfies:
\begin{equation}
\label{eq:ode_u}
    \partial_t u = -u + \omega \ast [f(u_{\textsc{NL}})-u_{\textsc{L}}].
\end{equation}
By Proposition~\ref{rmk:linear-upper-bound} and the assumption on $f$, for $t$ sufficiently large we have $f(u_{\textsc{L}}(t))= u_{\textsc{L}}(t)$. Hence, for $t$ sufficiently large and any $x\in \R^d$ it holds
\begin{equation}
\label{eq:taylor2}
    \begin{aligned}
        f(u_{\textsc{NL}}) - u_{\textsc{L}} 
        &= f(u+ u_{\textsc{L}}) - f(u_{\textsc{L}})\\
        &=\int_{u_{\textsc{L}}}^{u+u_{\textsc{L}}}f'(s)ds.
    \end{aligned}
\end{equation}

Taking the $L^p$ norm of the solution to \eqref{eq:ode_u} with initial condition $u_0 \in L^p(\R^d)$ written using Variation of Constant formula, exploiting \eqref{eq:taylor2}, we have:
\begin{equation}
    \begin{aligned}
        e^t\norm{u(t)}_{p} -\norm{u_0}_{p}& \leq  \norm{\omega}_1 \int_0^t e^s \left\|{\int_{u_{\textsc{L}}}^{u+u_{\textsc{L}}}f'(s)ds}\right\|_{p} \\
               &\leq  \norm{\omega}_1 \int_0^t e^s \norm{u( s)}_{p} ds .
    \end{aligned}
\end{equation}
Applying Gronwall's Lemma (see Lemma~\ref{lem::gronwall-lemma}), the statement follows.
\end{proof}

\section{Explicit solutions in the linear case}
\label{sec:linear}
In this section, we set the firing rate function to be linear, i.e. $f(s)= s $, and we deal with the following semi-linear neural field equation:
\begin{equation}
\label{eq:mean_field-lin}
\frac{\partial u}{\partial t}= - u + \omega\ast u + I,
\end{equation}
for a $T$-periodic input $I$.
The goal is to provide an explicit expression for its unique attractive periodic solution.
To lighten the notation, we will assume the periodicity of the external input $I$ to be $T=2\pi$.

\subsection{Decomposition of the solution}
% Assume now the kernel $\omega$ to be a real, even (i.e. $\omega(x) = \omega(-x)$ for all $x \in \R^d$) function. 
Let us define $a(\xi) \coloneqq \hat{\omega}(\xi)-1$, where $\hat{\omega}$ is the Fourier transform of the kernel $\omega$.
For a generic function $u \in L^p(\R^d)$ allowing Fourier transform, we have:
    \begin{equation}
    \label{eq:FT_op}
        \begin{aligned}
            Au(x) = & \int_{\R^d} a(\xi)\hat{u}(\xi)e^{2\pi i \langle x , \xi \rangle} d\xi,  \\
            e^{tA}u(x) = & \int_{\R^d} e^{ta(\xi)}\hat{u}(\xi)e^{2\pi i \langle x , \xi \rangle} d\xi,\\
        \end{aligned}
    \end{equation}
using the operator notation introduced in equation \eqref{eq:operator_notation} with $f(s)=s$. 

\begin{remark}
Due to the fact that $\omega$ is even, we have that $a(\xi)$ is real for any $\xi\in\R$. Hence, 
\begin{equation}
\label{eq:well_def_den}
     i\ell-a(\xi) \neq 0 \text{ for all } \ell \in \Z \text{ and } \xi \in \R.
\end{equation}  
Moreover, thanks to $\norm{\omega}_1 < 1$,  it holds $a(\xi) < - c$ with $c$ real positive constant.
\end{remark}

Since the input $I$ is $2\pi$-periodic, we can write it as the Fourier series
\begin{equation}
\label{ass:Fourier_Series}
    I(x, t) = \sum_{\ell \in \Z} I_\ell(x) e^{i \ell t}, \qquad \forall t\in \mathbb{R},
\end{equation}
here, we let
\begin{equation}
    \label{eq:Il}
I_{\ell}(x) = \frac{1}{2\pi} \int_0^{2\pi} I(x, \tau) e^{-i \ell \tau} \, d\tau.    
\end{equation}
We denote by $\hat I_\ell$ the Fourier transform in space of $I_\ell$. With these notations, we have the following.

\begin{proposition}
\label{prop:zero_term}
The initial condition $u_{0, I} \in L^p(\R^d)$ that determines the periodic orbit $u^\star$ reads
    \begin{equation}
     u_{0, I} = \sum_{\ell \in Z }  \int_{\R^d}  \frac{\hat{I_\ell}(\xi)e^{2\pi i \langle x , \xi \rangle}}{ i\ell-a(\xi)}  d\xi.
    \end{equation}
\end{proposition}

\begin{proof}
We consider the solution $u(t)$ of \eqref{eq:mean_field-lin} with initial condition $u_{0, I}$ and we focus on the non-homogeneous term of  $u(t)$ written using Variation of Constant formula. Taking advantage of the notation introduced in \eqref{eq:FT_op}, we have
\begin{equation}
\begin{aligned}
     &e^{tA}\int_0^t e^{-sA}I(s)ds\\
     &=e^{tA}g + \sum_{\ell \in Z } \left( \int_{\R^d} \frac{\hat{I_\ell}(\xi)e^{2\pi i \langle x , \xi \rangle}}{ i\ell-a(\xi)}   d\xi\right) e^{ i \ell t} \\
\end{aligned}
\end{equation}
with 
\begin{equation}
    g(x) \coloneqq - \sum_{\ell \in Z }  \int_{\R^d}  \frac{\hat{I_\ell}(\xi)e^{2\pi i \langle x , \xi \rangle}}{ i\ell-a(\xi)}  d\xi .
\end{equation}
Since $u(t)$ is periodic and coincides with $u^\star(t)$ if and only if $u_{0, I} + g = 0$, then  $u_{0, I} = -g$.
\end{proof}

\begin{theorem}
\label{prop:limit-cycle-char}
Let $u^\star$ be the globally attractive periodic solution of \eqref{eq:mean_field-lin}. The following decomposition holds:
 \begin{equation}
\label{eq:limit_cycle}
\begin{aligned}
    u^{\star}(x, t) 
     & = I(x, t) + \sum_{\ell \in \Z}(I_\ell\ast K_\ell)(x) e^{ i \ell t},
\end{aligned}
\end{equation} 
where $I_\ell$ is defined in \eqref{eq:Il}, and the Fourier transform of $K_\ell$ is given by
\begin{equation}
	\label{eq:Kl}
	\hat{K_\ell}(\xi) = \frac{  i \ell -\hat{\omega}(\xi)}{ i\ell+1-\hat{\omega}(\xi)}.
\end{equation}
\end{theorem}
\begin{proof}
It follows from the proof of Proposition \ref{prop:zero_term} that we can write the periodic attractive solution $u^{\star}$ of Theorem \ref{thm:attractive_limit_cycle} using the following expansion
    \begin{equation}
    \label{eq:series-sol}
        u^{\star}(x,t) = \sum_{\ell \in \Z} u^{\star}_\ell(x) e^{ i \ell t } \text{ for all } t,
    \end{equation}
    where the coefficients $u^{\star}_\ell(x)$ are given by: 
    \begin{equation}
    \label{eq:coeff_Il}
        u^{\star}_\ell(x) = \int_{\R^d} \frac{\hat{I_\ell}(\xi)}{ i \ell -  a(\xi)}e ^{2  \pi i \langle x , \xi \rangle}d\xi.
    \end{equation}
Since $a(\xi) = \hat{\omega}(\xi)-1$, we have:
\begin{equation}
        \label{eq:kernel_Kl}
        \frac{1}{ i\ell-a(\xi)}= 1 + \hat{K_\ell}(\xi)\text{ , }  \hat{K_\ell}(\xi) = \frac{  i \ell -\hat{\omega}(\xi)}{ i\ell+1-\hat{\omega}(\xi)}.
\end{equation}
Substituting \eqref{eq:kernel_Kl} in  \eqref{eq:coeff_Il}, the proof of the statement follows applying the Convolution theorem to the product  $\hat{I_\ell}\hat{K_\ell}$.
\end{proof}

\subsection{Inverting 1D kernels}
From now on we will consider $d = 1$ (i.e $\R^d = \R$), and in accordance with previous results on static external input  (\cite{tamekue2023mathematical, tamekue2024reproducibility}), we assume the kernel to be defined as a difference of Gaussians \footnote{The difference of Gaussians kernel, commonly used in visual processing, models neural interactions by encoding short-range excitation and long-range inhibition between neurons.}:
\begin{equation}
\label{eq:DoG_kernel}
            \omega (x) := (2\pi\sigma_1^2)^{-1}e^{ - \frac{\abs{x}^2}{2\sigma_1^2}}- k(2\pi\sigma_2^2)^{-1}e^{ - \frac{\abs{x}^2}{2\sigma_2^2}}
\end{equation}
with $ k \geq 1$, $0 \leq \sigma_1 < \sigma_2$, and $ \sigma_1 \sqrt{k} \leq \sigma_2$. For reasons that will be evident later on, we assume moreover that $\sigma_1^2/\sigma_2^2 \in \Q$.

\begin{remark}
    According, e.g., to \cite{tamekue2023mathematical}, letting $\sigma_2=\alpha \sigma_1$ for $\alpha > 1$, we have
    \begin{equation}
        \|\omega\|_1= 1- k+ 2\left[ k \left(\frac{\alpha ^2}{k}\right)^{\frac{1}{1-\alpha ^2}} -\left(\frac{\alpha ^2}{k}\right)^{\frac{\alpha ^2}{1-\alpha ^2}}  \right].
    \end{equation}
    As a consequence, the assumption $\|\omega\|_1<1$ is satisfied for $k$ and $\alpha$ sufficiently near to $1$. In particular, if $\alpha=\sqrt{2}$, this holds for $k\in [1,2)$.
\end{remark}

The Fourier transform of $\omega$ can be explicitly expressed as
\begin{equation}
\label{eq:omega_hat}
     \begin{aligned}
        \hat{\omega}(\xi) =& e^{-2\pi^2\sigma_1^2 \xi^2}- ke^{ - 2\pi^2\sigma_2^2\xi^2}.\\
    \end{aligned}
\end{equation}

The objective here is to compute the inverse transform of $\hat{K_\ell}(\xi)$ when $\omega$ is defined as \eqref{eq:DoG_kernel}. 
The idea is to apply the Residue theorem to compute the real integral of
\begin{equation}
\label{eq:meromorphic}
\hat{K_\ell}(\xi)e^{2\pi i x \xi}= \frac{(i \ell -\hat{\omega}(\xi))e^{2\pi i x \xi}}{ 1 + i \ell  - \hat{\omega}(\xi)},\quad \xi \in \R,  
\end{equation}
generalizing the strategy adopted in the case of static inputs in \cite{tamekue2023mathematical}.

By definition, the function \eqref{eq:meromorphic} consists of exponentials that can be extended into the complex plane $\C$. We then introduce the set of its simple poles
\begin{equation}
\label{def:poles}
P_\ell \coloneqq \{z \in \C \mid 1+ i\ell-\hat{\omega}(z) = 0\}.
\end{equation}
As noted in \cite{tamekue2023mathematical} (see also \cite{berenstein2012complex}), given that $1+i\ell -\hat{\omega}(z)$ is a complex exponential polynomial, the assumption $\sigma_1^2/\sigma_2^2 \in \Q$ implies that $P_\ell$ is discrete and countable.

Denoting by $\bar{z}$ the complex conjugate of $z\in \C$, the set of poles has the following symmetries.

\begin{proposition}
\label{prop:poles}
Let $z\in P_\ell$, $\ell\in\Z$. Then,
\begin{equation}
    -z \in P_\ell
    \qquad\text{and}\qquad
    \bar{z}, -\bar{z} \in P_{-\ell}.
\end{equation}
\end{proposition}

The proof of this proposition depends on the explicit form of the kernel $\omega$ introduced in \eqref{eq:DoG_kernel}. Details are provided in Appendix~\ref{app:proofs} (see Proposition \ref{prop:app_poles}). 

For the static case ($\ell = 0$), it directly follows that the presence of one pole automatically determines the presence of three additional ones.

\begin{corollary}
\label{cor:0-poles}
Assume $\ell = 0$. If $z \in P_0 $ , then $\bar{z}, -\bar{z}, -z \in P_0$.
\end{corollary}

We then compute $K_\ell$, leveraging on the symmetries of Proposition \ref{prop:poles}.
More precisely, we focus on the poles that lie in the positive quadrant 
\begin{equation}
\label{eq:P+}
    \mathbf{P}^{+}_\ell \coloneqq \{ z \in P_\ell\cup P_{-\ell} 
        \mid \Re z > 0 , \Im z>0\};
\end{equation}
in particular, we can write 
\begin{equation}
    \mathbf{P}_\ell^+ = \bigcup_{n=0}^{+\infty} \bigg(\mathbf{P}_\ell^+\cap \{|z|<R_n\}\bigg),
\end{equation}
where $R_n = (M n + \varepsilon)^{1/\alpha}$, for appropriate positive constants $M, \varepsilon, \alpha$. 
Since $\mathbf{P}_\ell^+$ is discrete, we pick the enumeration\footnote{To lighten the notation, we do not explicit the dependence of $z_j$ on $\ell$.} 
\begin{equation}
	\label{eq:poles++}
	\mathbf{P}_\ell^+=\{z_0,z_1,\ldots\}
\end{equation}
 such that $\Im z_j\le \Im z_{j+1}$ and, if $\Im z_j = \Im z_{j+1}$, it holds ${\Re z_j}\le {\Re z_{j+1}}$.

\begin{figure}[tbh]
    \centering
    \includegraphics[width=.45 \textwidth]{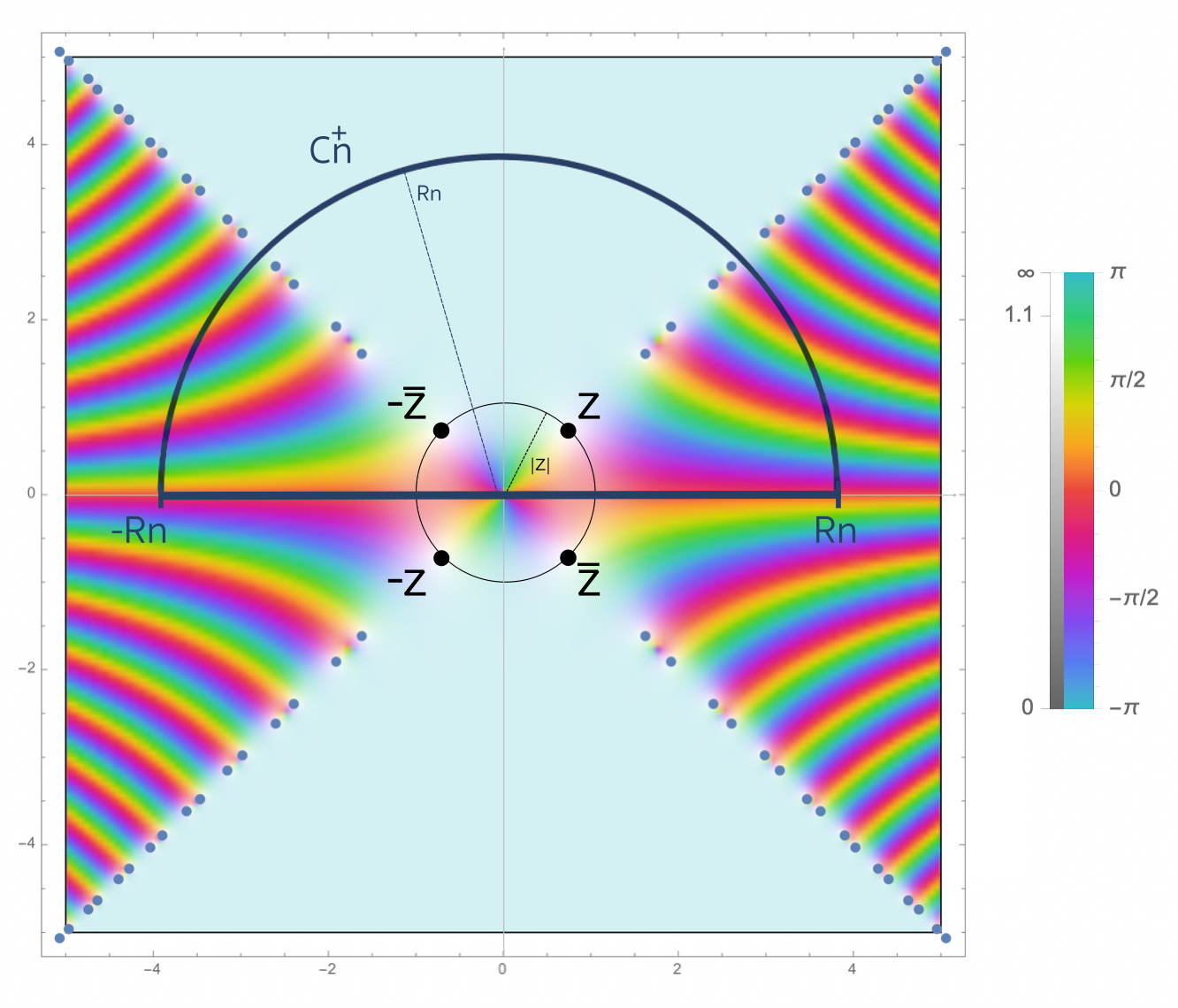}
\caption{Visualization of the meromorphic continuation of \eqref{eq:meromorphic} for $\ell = 0$ in the complex plane $\mathbb{C}$. The parameters for the kernel $\omega$ are set to be $k = 1$, $2\pi^2\sigma_1^2 = 1$, and $2\pi^2\sigma_2^2 = 2$. The color indicates the argument of the function values in the complex plane, with the graph becoming paler as the absolute value increases. The poles (zeros of $1+\hat{\omega}(z)$) are highlighted in blue: for any $z \in P_0$, the points $-z, -\bar{z}, \bar{z} \in P_0$ (see Corollary \ref{cor:0-poles}). Finally, the thick black line represents the integration path used in Proposition \ref{prop:K_series}.}

\label{fig:example_poles}
\end{figure}

\begin{proposition}
\label{prop:K_series}
The kernel $K_\ell(x)$ defined in \eqref{eq:Kl} can be recast, if $\ell \neq 0$, as
\begin{equation}
{K}_\ell(x) =  
    \sum_{j=0}^{+\infty}\frac{2\pi e^{-2\pi \lvert x \rvert \Im z_{j}}}{c_j}e^{-\sgn(\ell)i(2\pi \Re z_j\lvert x\rvert+\phi)},
\end{equation}
or, if $\ell = 0$, as 
\begin{equation}
{K}_\ell(x) = \sum_{j=0}^{+\infty}\frac{4\pi e^{-2\pi \lvert x \rvert \Im z_{j}}}{c_j}\cos(2\pi \lvert x \rvert \Re z_j+ \phi),
\end{equation} 
where $z_j \in \mathbf{P}_\ell^+$, and we denote with $c_j := \hat{\omega}'(z_j)$, and with $\phi:=\arctan\left(\Re c_j/\Im c_j\right)$.
 \end{proposition}

\begin{proof}
We outline the proof here, with further details provided in Lemma \ref{lem:K_series}. The key point is using the symmetries of the poles identified in Proposition \ref{prop:poles} to apply the Residue Theorem.

We start by observing that
\begin{equation}
	K_\ell = 
	\begin{cases}
		2K_0^2 & \text{if } \ell=0\\
		\sgn(\ell)K^1_\ell+K_\ell^2 & \text{if } \ell \neq 0
	\end{cases}
\end{equation}
where the Fourier transforms of ${K}^1_\ell$ and ${K}^2_\ell$ are 
\begin{equation}
\label{eq:kernels_1_2}
\begin{aligned}
\hat{K}^1_\ell(\xi) &= \frac{i |\ell|}{(1+i\ell -\hat{\omega})(1-i\ell -\hat{\omega})(\xi)} \\
\hat{K}^2_\ell(\xi) &= \frac{\hat{\omega}(\xi)(\hat{\omega}(\xi)-1) +\ell^2}{(1+i\ell -\hat{\omega})(1-i\ell -\hat{\omega})(\xi)}.
\end{aligned}
\end{equation}
Observe that the set of poles of the meromorphic continuation of $K_\ell^\nu$, $\nu= 1, 2$, are exactly $P_\ell\cup P_{-\ell}$.

To compute the inverse Fourier transforms of \eqref{eq:kernels_1_2}, we use the Residue Theorem. More precisely, we will show that
\begin{equation}
	\int_{\R} \hat{K}_\ell^\nu (\xi)e^{2\pi i x\xi}\,d\xi = \lim_{n\to +\infty} \int_{\Gamma_n} \hat{K}_\ell^\nu (z)e^{2\pi i xz}\,dz, 
\end{equation}
for a particular choice of contours $\Gamma_n\subset \mathbb{C}$.
By the Residue Theorem, we will then compute the integrals on the l.h.s.~by computing the residues of $\hat{K}_\ell^\nu$ at the poles contained in $\Gamma_n$.

An example of our choice of integration path is shown in Figure \ref{fig:example_poles} for $\ell = 0$. Specifically, for $n \in \N$, we define $\Gamma_n := [-R_n, R_n] \cup C_n^+$, where $R_n = (M n + \varepsilon)^{1/\alpha}$, with $M, \varepsilon, \alpha$ as appropriate constants, $ [-R_n, R_n] \subseteq \R$, and $C_n^+ = \{z = R_n e^{i \theta} \in \C \mid \theta \in [0, \pi]\}$.

For $\ell \in \Z, \ell \neq 0$, we obtain (see Lemma \ref{lem:K_series}):
\begin{equation}
\label{eq:K1&K2}
\begin{aligned}
K_{\ell}^1(x) &= -2\pi i\sum_{j=0}^\infty \frac{e^{-2\pi x \Im z_j}}{\abs{c_j}} \sin(2\pi \abs{x} \Re z_j + \phi), \\
K_{\ell}^2(x) &= 2\pi \sum_{j=0}^\infty \frac{e^{-2\pi x \Im z_j}}{\abs{c_j}} \cos(2\pi \abs{x} \Re z_j + \phi),
\end{aligned}
\end{equation}
where $\phi = \tan^{-1} \left( {\Re(c_j)}/{\Im(c_j)} \right)$ and $c_j = \Re z_j + i \Im z_j$. The thesis follows by Euler's formula for the complex exponential.

Finally, for $\ell = 0$, we apply the same reasoning, noting that $\hat{K}^1_0 = 0$. The thesis follows by considering $2\hat{K}^2_0$.
\end{proof}

The following result is a direct consequence of Theorem \ref{thm:K0}.

 \begin{theorem}
\label{thm:K-approx}
Considering the enumeration \eqref{eq:P+}, we have:
    \begin{multline}
       \label{eq:kernel}
       \frac{ e^{2\pi \lvert x \rvert \Im z_{0}}}{2\pi/\abs{c_{0}}}{K}_\ell(x) 
       \\
       =  \begin{cases}
    e^{-\sgn(\ell)i(2\pi \Re z_0\lvert x\rvert+\phi)} +\mathcal{O}(1/x) & \ell \neq 0\\[.5em]
    2\cos(2\pi \lvert x \rvert \Re z_{0}+ \phi) +\mathcal{O}(1/x) &  \ell = 0.
\end{cases}
\end{multline}
\end{theorem}

\begin{remark}
Although we restricted to the case $T=2\pi$, the results of this section hold for any $ T > 0 $, by replacing $\ell$ with $\ell (2\pi / T)$ in the definition of $K_\ell$. In particular, this changes the poles appearing in Proposition~\ref{prop:K_series} and in Theorem~\ref{thm:K-approx}.
\end{remark}

\section{Application to visual processing}
\label{sec:application}

As stated in the introduction, neural field equations provide a tool for understanding the collective behavior of large populations of neurons, and they are particularly used in the context of visual perception. Here we are interested in understanding the solution concerning localized flickering and geometric visual stimuli as external inputs. In this context, the external stimulus $I$ represents cortical input, a deformed version of retinal stimulus due to the retino-cortical map.

\subsection{Retino-cortical map}

\begin{figure*}[tbh]
    \centering
\begin{subfigure}[b]{0.22\textwidth}
         \centering
         \includegraphics[width=\textwidth]{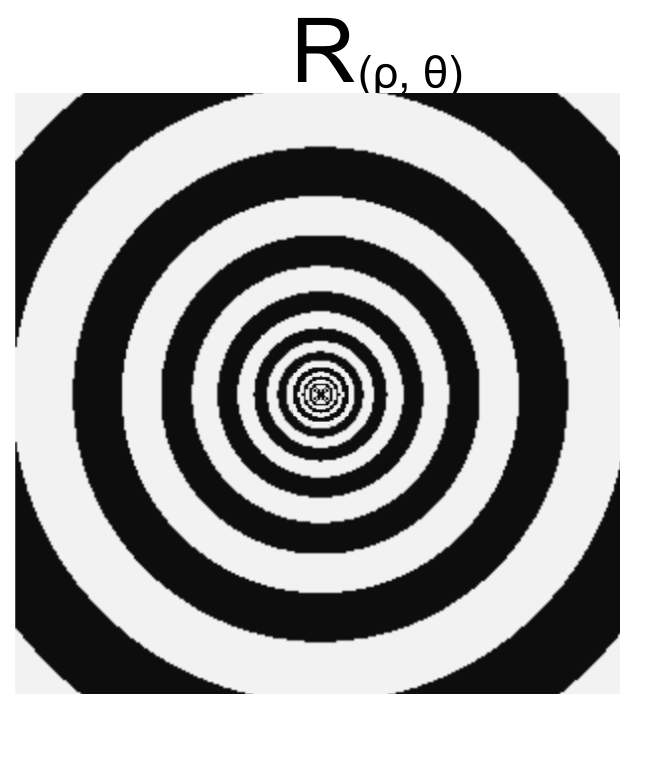}
         \caption{}
\end{subfigure}
\begin{subfigure}[b]{0.23\textwidth}
         \centering
         \includegraphics[width=\textwidth]{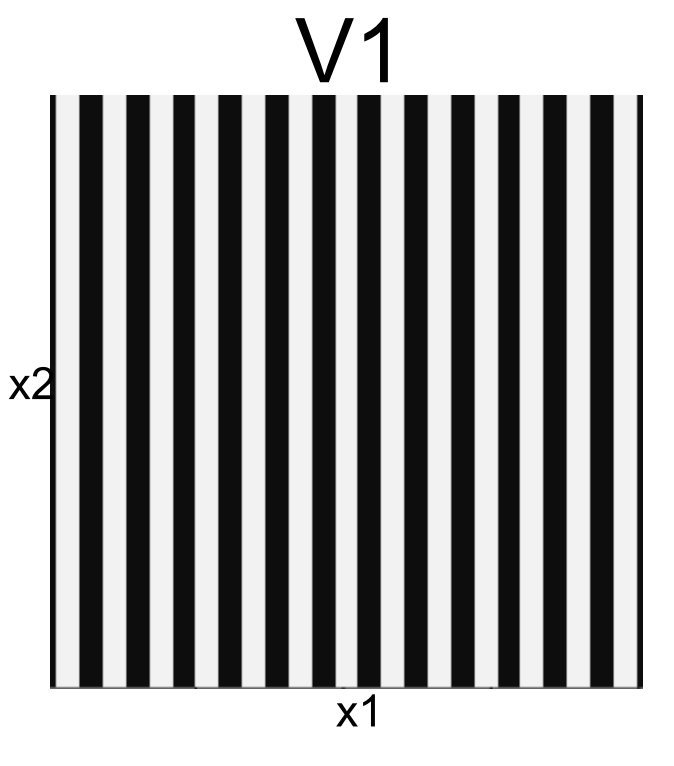} 
         \caption{}
\end{subfigure}
\begin{subfigure}[b]{0.23\textwidth}
         \centering
         \includegraphics[width=\textwidth]{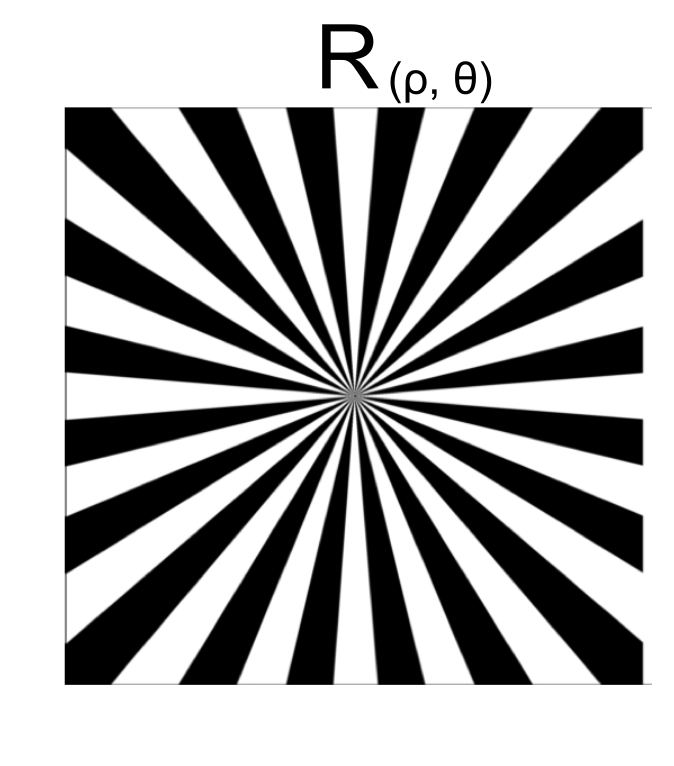} 
         \caption{}
\end{subfigure}
\begin{subfigure}[b]{0.23\textwidth}
         \centering
         \includegraphics[width=\textwidth]{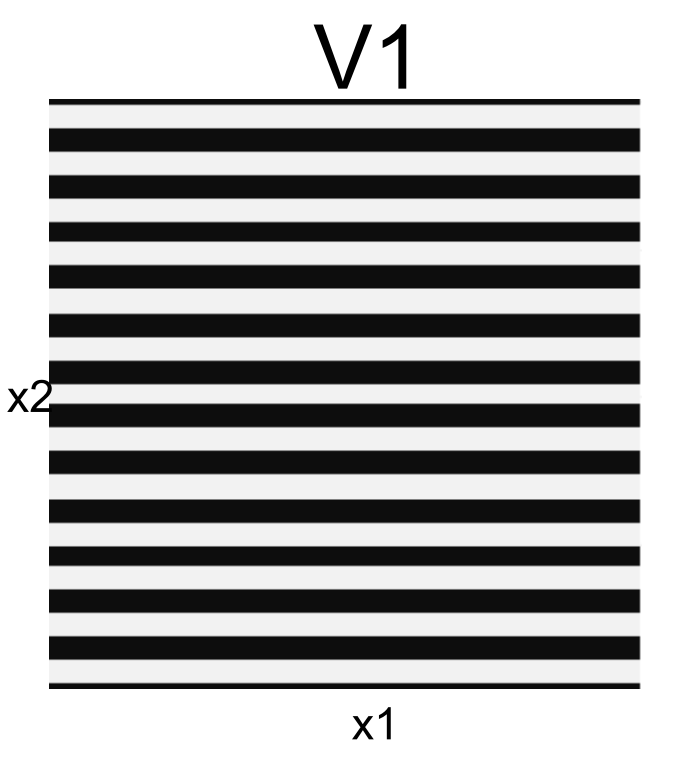} 
         \caption{}
\end{subfigure}\captionsetup{justification=justified, singlelinecheck=false, width=\linewidth}
\caption{Binary stimulus representation using \eqref{eq:binary}. (a) Display of the circular pattern in retinal coordinates (polar coordinates) of the visual field. (b) Display of the circular pattern in V1 coordinates via the complex logarithm map defined in \eqref{eq:retino-cortical-map}. This corresponds to the binary version of $I_C(x)= \cos(2\pi\zeta x_1)$. (c) Display of the radial pattern in retinal coordinates of the visual field. (d) Display of the radial pattern in V1 coordinates. This corresponds to the binary version of $I_R(x)= \cos(2\pi\zeta x_2)$. }
\label{fig:stimulus-representation}
\end{figure*}

The propagation of the visual signal from the retina to the primary visual cortex V1 is performed through the optical nerve. This propagation respects the topographic organization of the retina (\cite{tootell1982deoxyglucose}): adjacent locations on the retina project to adjacent neurons in V1 by a ``continuous'' point-by-point relationship. This induces a topographic map, referred to as the retino-cortical mapping. 

Formally, this map is shown to be well approximated, close to the center of the field of view, by the complex logarithmic function (e.g., see \cite{schwartz1977spatial}). More precisely, let $(\rho , \theta) \in (0, \infty)\times [0, 2\pi)$ denote the polar coordinates in the visual field (or in the retina) and $x\coloneqq(x_1, x_2)\in \R^2$ denote the cortical coordinates in V1. The analytic expression is given by:
\begin{equation}
\label{eq:retino-cortical-map}
\begin{aligned}
      \textsc{R} \cong \mathbb R^2 && \longrightarrow && \textsc{V1} \cong \mathbb R^2 \\
         (\rho, \theta) && \longmapsto && (x_1, x_2)
\end{aligned}
\end{equation}
where $(x_1, x_2):= (\log \rho, \theta)$. {We refer to \cite{bressloff2001geometric} for more details on this topic.}

\subsection{Binary representation and geometric stimuli}

It is common in psychophysical experiments, such as those performed in \cite{mackay1957moving, mackay1957some, billock2007neural}, to examine how flicker interacts with circular and radial patterns. These geometric stimuli can be represented in cortical coordinates $x \coloneqq (x_1, x_2) \in \R^2$ as follows:
\begin{equation}
I_C(x) = \cos(2\pi\zeta x_1), \hspace{ 0.5cm} I_R(x) = \cos(2\pi\zeta x_2)
\end{equation}
where $\zeta \in \R^+$. The choice of this representation for circular $I_C$ and radial $ I_R$ patterns is motivated by the analogous geometric hallucinatory patterns described in \cite{bressloff2001geometric} and \cite{ermentrout1979mathematical}.

To visualize $I_C$ and $I_R$ in terms of images, it is common to represent them as contrasting black and white regions, as displayed in Figure \ref{fig:stimulus-representation}. Specifically, we define the binary pattern $b_{I}$ of a function $I: \R^2 \rightarrow \R$ as
\begin{equation}
\label{eq:binary}
b_{I}(x) = \begin{cases}
0 & \text{ if } I(x) > 0 \text{ (black color)}\\
1 & \text{ if } I(x) \leq 0 \text{ (white color)}.\\
\end{cases}
\end{equation}
This representation highlights the zeros level-set $\{x\in\R^2: I(x)=0\}$ of $I$,  which delineates the boundary between different regions in the image generated by the function $I$, thus outlining the shapes within the image.

\subsection{Localized flickering input}
From a perceptual point of view, a possible way to model localized flickering behavior is by using a sinusoidal input. Typically, localized flickering information is concentrated either in the center of the visual field or its periphery, therefore, we consider the following representation in cortical coordinates of V1 
\begin{equation}
\label{eq:input-flicker}
    I(x, t) \coloneqq H(x_1)\cos(\lambda t) 
\end{equation}
with $ x = (x_1, x_2)\in \R^2, t \in \R, \lambda \in [0, +\infty)$ and where $H$ is the Heaviside step function 
\begin{equation*}
H(x_1)= \begin{cases}
    1 \quad \text{ if } x_1 \geq 0\\
    0 \quad \text{ if } x_1 < 0\\
\end{cases}.
\end{equation*}
The role of the Heaviside function is to localize the sinusoidal flicker either in the periphery (by taking $H(x_1)$) or in the center of the visual field (considering $H(-x_1)$).

In this Section, we study the effect of stimulus \eqref{eq:input-flicker} in the case where the firing rate function $f$ is linear. We recall that, in Section \ref{sec:general}, we showed that the solution for a general nonlinear $f$ is well-approximated by the linear case under appropriate assumptions on the input or on $f$.

%In Section \ref{sec:general}, we observed that for a small input in $L^\infty$ norm the solution of \eqref{eq:mean_field} with a nonlinear function $f$ is approximated by the solution with a linear $f$. Hence, we focus on studying the linear case, building upon the results from Section \ref{sec:linear}. 

We introduce the following notation
\begin{multline}
    \label{eq:circ-sol}
    V_{\lambda}^\ell(x,t) \coloneqq
    \\
    -\frac{\sgn(x_1) }{\abs{c_{0}}\abs{z_0}}e^{-2\pi \Im{z_0} \abs{x_1}}\cos(2\pi \Re{z_0}\abs{x_1}-\lambda t+\Phi);
\end{multline}
here, $z_0$ is the smallest pole in $\mathbf{P}_\ell^+$, according to the enumeration \eqref{eq:poles++}, $c_0=\omega'(z_0)$, and $ \Phi = \tan^{-1}\left(\Re c_0/\Im c_0\right) + \tan^{-1}\left({\Re{z_0}}/{\Im{z_0}}\right)$.

\begin{proposition}
\label{prop:flickering_equation}
Let us consider $I(x, t) = H(x_1)\cos(\lambda t)$, and assume $f(s)= s$. Then, solutions of \eqref{eq:mean_field-lin} converge toward the attractive periodic state: 
\begin{equation}
\label{eq:flicker-sol}
\begin{split}
 u^\star&(x,t)\\
 &= I(x, t) + V_\lambda^1(x,t) + \mathcal{O}\left(\frac{e^{-2\pi |x_1| \Im{z_0}}}{x_1} \right), 
\end{split}
\end{equation} 
where $V_\lambda^1$ is defined in \eqref{eq:circ-sol} with $\ell = 1$.
\end{proposition}

\begin{proof}
We compute $u^\star$ relying on \eqref{eq:limit_cycle} and \eqref{eq:kernel}, by replacing $ \ell$  with  $\ell \lambda$  since the time periodicity of $I$ is $ 2\pi/\lambda$.

We notice that, for $\ell = 0$ we have $(I\ast K_0)(x_1) =0$.
On the other hand, for $\ell \neq 0$, the Fourier coefficients are given by $ I_\ell(x_1)= H(x_1)/2 \text{ if } \ell = \pm 1$, otherwise $I_\ell(x_1)= 0$.
Then, we get the statement computing
\begin{equation}
\begin{split}
    &(K_1\ast H)(x_1)  \frac{e^{i\lambda t}}{2} +(K_{-1}\ast H)(x_1)\frac{e^{-i\lambda t}}{2} \\
    &=\cos(\lambda t)(K_1^2\ast H)(x_1) + i \sin (\lambda t)(K_1^1\ast H)(x_1)  \\
    &=V_\lambda^1(x,t),
\end{split}
\end{equation}
where the last equality has been obtained integrating by parts and rearranging according to trigonometric formulas. 
\end{proof}

\begin{figure*}[tbh]
\centering
\begin{subfigure}[b]{0.24\textwidth}
         \centering
         \includegraphics[width=\textwidth]{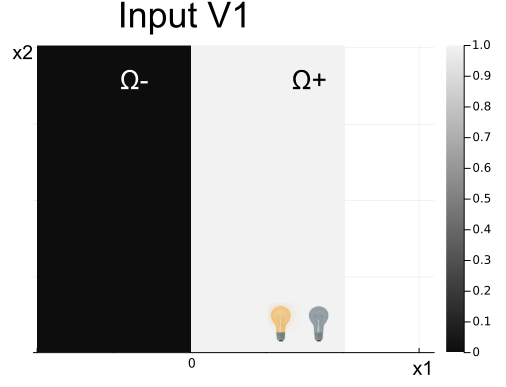}
         \caption{}
\end{subfigure}
\begin{subfigure}[b]{0.24\textwidth}
         \centering
         \includegraphics[width=\textwidth]{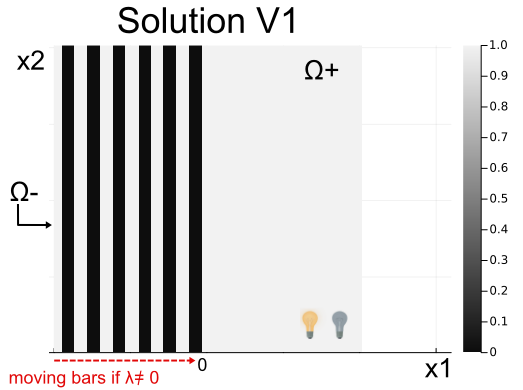} 
         \caption{}
\end{subfigure}
\begin{subfigure}[b]{0.24\textwidth}
         \centering
         \includegraphics[width=\textwidth]{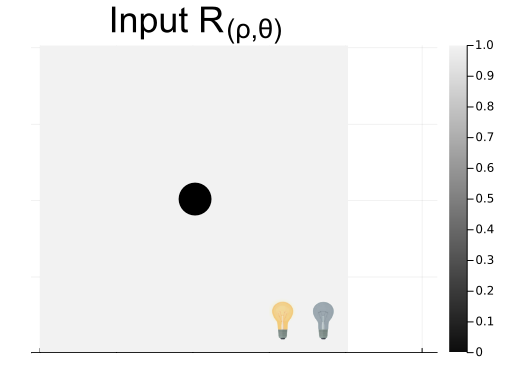} 
         \caption{}
\end{subfigure}
\begin{subfigure}[b]{0.24\textwidth}
         \centering
         \includegraphics[width=\textwidth]{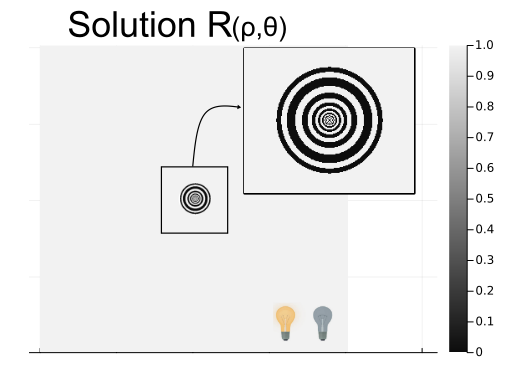} 
         \caption{}
\end{subfigure}
      \caption{ Binary representation of the input \eqref{eq:input-flicker} and of the corresponding solution \eqref{eq:flicker-sol} (or \eqref{eq:heaviside-sol} if $\lambda=0$). Bulbs in the white region indicate the presence of flickering lights if $\lambda \neq 0$. The kernel parameters for \eqref{eq:DoG_kernel} are set to be $k= 1$, $\sigma_1 = 1/\sqrt{2}\pi$ and $\sigma_2 = 2/\sqrt{2}\pi$. 
       (a) Display of \eqref{eq:input-flicker} in cortical coordinates. If $\lambda \neq 0$, $\Omega^+$ switches from black to white due to the presence of flicker. (b) Display of \eqref{eq:heaviside-sol} in cortical coordinates. If $\lambda \neq 0$, the vertical stripes are moving stripes to the right with frequency $\lambda$. (c) Display of the input \eqref{eq:input-flicker} in retinal coordinates. (d) Display of the corresponding solution \eqref{eq:heaviside-sol} in retinal coordinates. 
      }
\label{fig:heaviside-input}
\end{figure*}

We observe that when $ \lambda = 0$ the input $ I $ defined in \eqref{eq:input-flicker} transitions from describing a flickering region to one characterized by static localized information. In this case, we can compute the solution similarly.

\begin{corollary}
\label{prop:zero_term_solution}
Let us consider
 $I(x) = H(x_1)$, and assume 
 $f(s)= s$. Then, the solutions of \eqref{eq:mean_field-lin} converge toward the attractive stationary state: 
\begin{equation}
\label{eq:heaviside-sol}
\begin{aligned}
    u^\star&(x) = I(x) + V_0(x)+ \mathcal{O}\left(\frac{e^{-2\pi |x_1| \Im{z_0}}}{x_1} \right) \\
\end{aligned}
\end{equation}
where $V_0(x)\coloneqq V_0^0(x,t)$ defined in \eqref{eq:circ-sol} with $\ell=0$.
\end{corollary}

The input $I$ defined in \eqref{eq:input-flicker} corresponds to peripherally localized information in the visual field, subjected to flickering (if $\lambda = 0$, we only have localized information). In cortical coordinates, the function $I$ partitions the corresponding image into two complementary sets:
\begin{equation}
\label{eq:w+}
\begin{aligned}
    \Omega^+ &\coloneqq \{x \in \R^2 :\: x_1\geq 0\},\\
    \Omega^- &\coloneqq \{x \in \R^2 :\: x_1< 0\}.\\
\end{aligned}
\end{equation}
The set $\Omega^+$ identifies the region of the image abundant in localized information, while $\Omega^-$
corresponds to the region characterized by the absence of visual input. 

We present a visual depiction of the input \eqref{eq:input-flicker} and of the corresponding solution $u^\star$, at a fixed time $\bar t$, in Figure~\ref{fig:heaviside-input}. Images (a) and (c), represents the input in cortical and retinal coordinates, respectively. The cortical and retinal representations of the solution are presented in images (b) and (d), respectively.
We notice $u^\star$ partitions its domain in accordance with the input: $\Omega^+$ remains unaltered, while $\Omega^-$ undergoes a transformation, exhibiting patterns characterized by moving vertical stripes of uniform width, a consequence of the cosine that appears in the second term on the l.h.s. of \eqref{eq:flicker-sol}. For large values of $x_1$, the exponential term ensures that $u^\star$ remains positive across $\Omega^+$ while oscillating within $\Omega^-$. Near $x_1=0$, where this reasoning might not apply, small oscillations may occur in $\Omega^+$. Here, the outcome depends on the principal pole, and consequently on parameters $\sigma_1$, $\sigma_2$ and $\lambda$.

We achieve similar results by examining \eqref{eq:heaviside-sol} when $\lambda = 0$  with the difference that vertical stripes remain static rather than moving. 

Finally, it is worth noticing that we would obtain the same results, if we were to consider $H(x_2)$ instead of $H(x_1)$. The only distinction would be at the level of the emergent pattern: instead of vertical, we would have horizontal stripes corresponding to radial patterns (see Figure \ref{fig:stimulus-representation}, images (c) and (d)). 

\subsubsection{Biasing the width of the stripes}
\label{sec:biasing}

The size of the vertical stripes in \eqref{eq:flicker-sol} depends on the principal pole $z_0 \in \mathbf{P}_1^+$, and consequently on the parameters defining $\omega$, namely $\sigma_1$, $\sigma_2$, and the flicker frequency $\lambda$ of the input. In particular, we are interested in understanding how the principal pole's real and imaginary parts vary when the flicker frequency increases. 
Indeed, by Proposition~\ref{prop:flickering_equation} the width of the induces stripes is approximately $1/(2\Re z_0)$.

\begin{proposition}
\label{prop:Rez_Imz}
The principal pole $z_0 \in \mathbf{P}_1^+$ satisfies
\begin{multline}
\label{eq:asymp_estimate}
\Re z_0 = \frac{1}{4 \sigma_2 \sqrt{2\log(\lambda/k)}} + \mathcal{O}\left({{\log^{-3/2}(\lambda/k)}}\right), \\
\Im z_0 = \frac{\sqrt{\log(\lambda/k)}}{\sqrt{2}\pi \sigma_2}+ \mathcal{O}\left({{\log^{-3/2}(\lambda/k)}}\right). 
\end{multline}
\end{proposition}

\begin{proof}
First, observe that the elements $z \in \mathbf{P}_1^+$, as the flicker frequency $\lambda$ varies, satisfy one of the following equations:
\begin{equation}
\label{eq:zeros}
1 + i\lambda = \hat{\omega}(z)
\qquad\text{or}\qquad
1 - i\lambda = \hat{\omega}(z),
\end{equation}
where $\hat{\omega}$ is defined in \eqref{eq:omega_hat}. 

We start by claiming that for any $z\in \mathbf{P}_1^+$ it holds
\begin{equation}
    \label{eq:approx-omega}
    \hat\omega(z) = -ke^{-2\pi\sigma_2^2z^2}(1+R(z))),
\end{equation}
where $R(z)=\mathcal{O}(\lambda^{-\gamma})$ with $\gamma\coloneqq 2\frac{\sigma_2^2-\sigma_1^2}{\sigma_1^2}>0$. Indeed, taking the modulus in \eqref{eq:zeros} and using the fact that $\sigma_1<\sigma_2$ we obtain
\begin{equation}
    1+\lambda^2 \le (1+k)e^{-2\pi\sigma_1^2(\Re z^2-\Im z^2)}.
\end{equation}
This yields 
\begin{equation}
    (\Re z)^2-(\Im z)^2 \le -\frac{1}{2\pi\sigma_1^2}\log\frac{1+\lambda^2}{1+k}.
\end{equation}
This allows to estimate 
\begin{equation}
    \left| R(z) \right| = k^{-1} \left|e^{-2\pi (\sigma_1^2-\sigma_2^2)z^2}\right| \le k^{-1} \left(\frac{1+\lambda^2}{1+k}\right)^{-\frac{\gamma}2},
\end{equation}
proving the claim.

We now replace the expression \eqref{eq:approx-omega} in \eqref{eq:zeros}, obtaining
\begin{equation}
    e^{-2\pi\sigma_2^2z^2} = -\frac{1+i\lambda}{k} \left( 1+\mathcal{O}(\lambda^{-\gamma})\right).
\end{equation}
Taking the (complex) logarithm\footnote{One has to pay attention to the fact that the complex logarithm is a multi-valued function.} we obtain that $z=w+i2\pi\ell+\mathcal{O}(\lambda^{-\gamma/2})$, where $\ell\in\Z$ and
\begin{equation}
    w^2 = -\frac{1}{2\pi\sigma_2^2} \left( \log\frac{\sqrt{1+\lambda^2}}{k} + i \arctan \lambda \right).
\end{equation}
One can then solve the above in terms of $\Re w$ and $\Im w$. This yields a fourth degree equation, whose solution can be approximated via the Taylor expansion of the square root for $\lambda \gg 1$, yielding the remainder term $\mathcal{O}({\log^{-3/2}(\lambda/k)})$ appearing in the statement. The statement follows by observing that, in order for $z$ to be the principal pole $z_0$ (i.e., the solution of \eqref{eq:zeros} with smallest real and imaginary part), one has to choose $\ell = 0$.
\end{proof}

\begin{figure}[tbh]
\centering
         \includegraphics[width=0.45\textwidth]{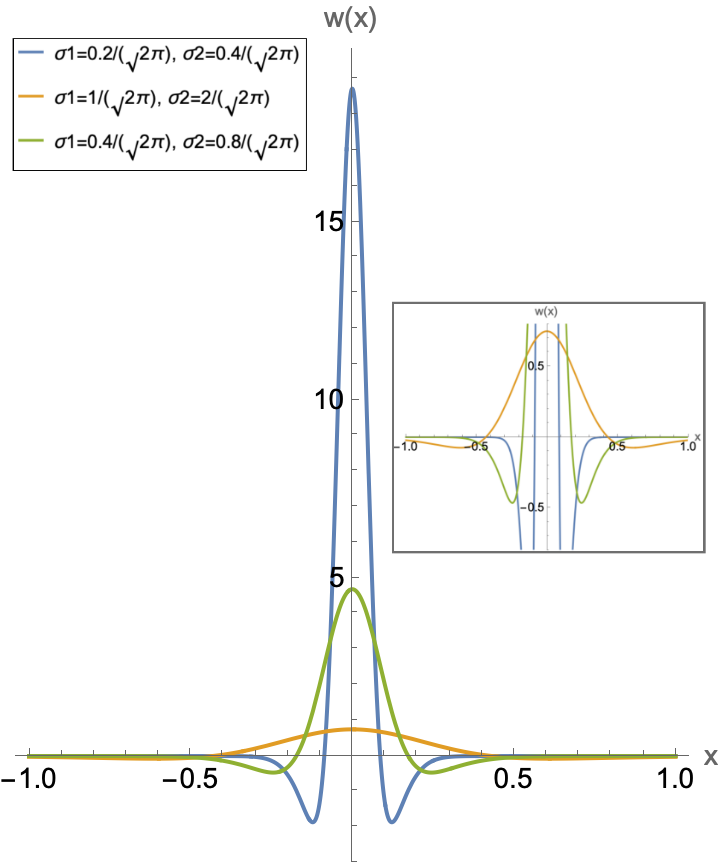}
      \caption{ Display of the kernel ${\omega}(x)$ defined by equation \eqref{eq:DoG_kernel}, with $k = 1$ and various values of $(\sigma_1, \sigma_2)$ as specified in \eqref{eq:sigma-paramters}. The box on the right provides a zoomed-in view of the central region near the origin to better illustrate the behavior of the three functions.
      }
\label{fig:kernels}
\end{figure}
 From this proof, particularly from \eqref{eq:approx-omega}, we notice the dominance of the inhibitory part of the kernel \eqref{eq:DoG_kernel} in influencing, along with flicker frequency, the behavior of the poles. This emphasizes the importance of inhibition in neural models, consistent with the observations in \cite{tamekue2024reproducibility}, where inhibition, interpreted through the firing rate function, plays a crucial role in shaping neural activity and enabling the model to accurately replicate visual phenomena.

To visualize this parameters dependence, we take $\lambda \in [2, 100]$, while for the kernel \eqref{eq:DoG_kernel} we set $k=1$ and consider the following couples: 
\begin{equation}
\label{eq:sigma-paramters}
\begin{aligned}
    1) :& \quad  (\sigma_1, \sigma_2) = (0.2/\sqrt{2}\pi, 0.4/\sqrt{2}\pi)\\
    2) :&  \quad (\sigma_1, \sigma_2) = (1/\sqrt{2}\pi, 2/\sqrt{2}\pi)\\
    3) :& \quad (\sigma_1, \sigma_2) = (0.4/\sqrt{2}\pi, 0.8/\sqrt{2}\pi).\\
\end{aligned}
\end{equation}
Figure \ref{fig:kernels} illustrates how the kernel $ \omega$ changes with the different values of $(\sigma_1, \sigma_2)$ specified. We note that these values satisfy $\|\omega\|_1 < 1$.

Then, for every pair $(\sigma_1, \sigma_2)$ of \eqref{eq:sigma-paramters}, we numerically compute with Mathematica the principal pole $z_0 = z_0(\lambda)$, varying $\lambda$ in the range $[2, 100]$.
We then plot the real part $\Re z_0(\lambda)$, the imaginary part $\Im z_0(\lambda)$, and the size of the vertical stripes $1/(2\Re z_0(\lambda))$ for each pair $(\sigma_1, \sigma_2)$. The results are displayed in Figure \ref{fig:biasing}, where we compare the approximating functions of Proposition \ref{prop:Rez_Imz} (illustrated with dashed lines), with the numerical computations (depicted with solid lines). 
In general, as the frequency increases, $\Re z_0(\lambda)$ decreases (Figure \ref{fig:biasing}, image (a)), while $\Im z_0(\lambda)$ increases (Figure \ref{fig:biasing}, image (b)). Consequently, this affects the width of the bars, which is given by $1/(2\Re z_0)$; Figure \ref{fig:biasing}, image (c) illustrates its increase as $\lambda$ grows.

\begin{figure*}[tbh]
\centering
\begin{subfigure}[b]{0.32\textwidth}
         \centering
         \includegraphics[width=\textwidth]{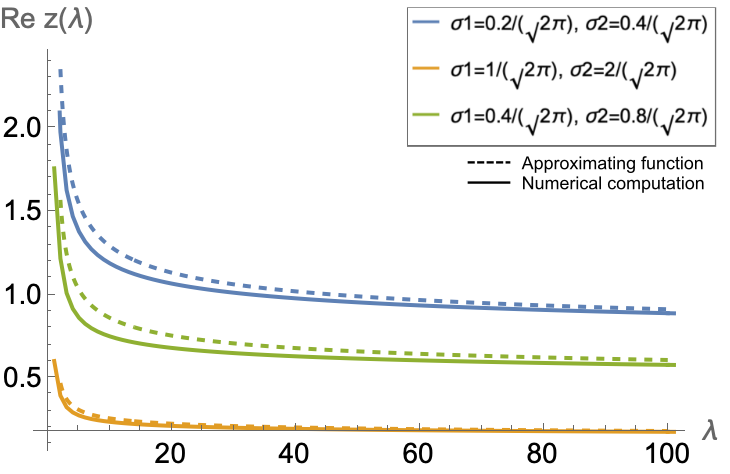} 
         \caption{}
\end{subfigure}
\begin{subfigure}[b]{0.32\textwidth}
         \centering
         \includegraphics[width=\textwidth]{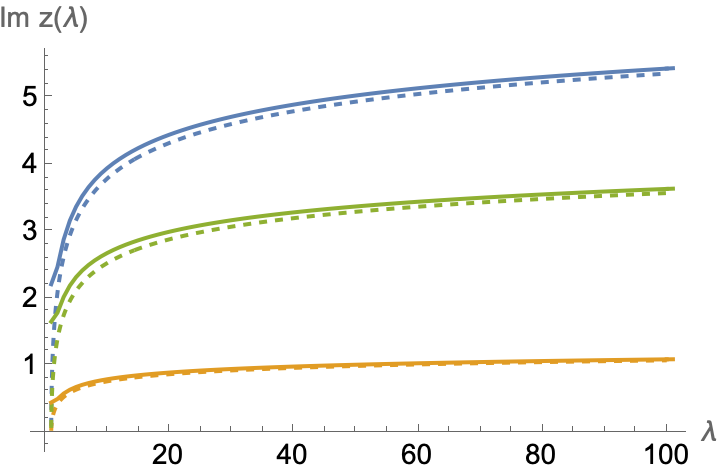} 
         \caption{}
\end{subfigure}
\begin{subfigure}[b]{0.32\textwidth}
         \centering
         \includegraphics[width=\textwidth]{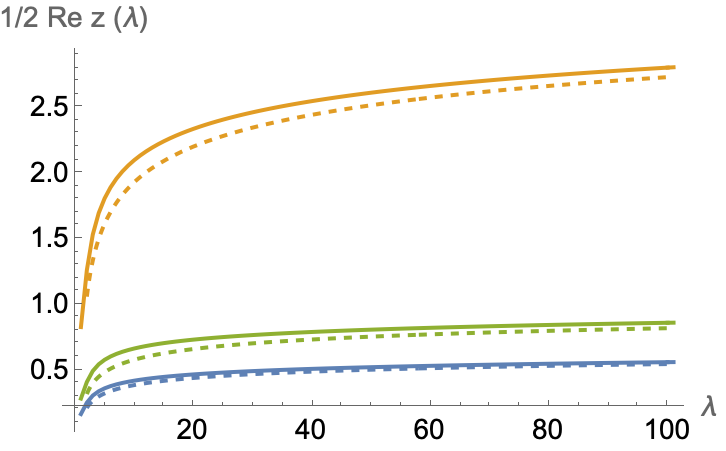} 
         \caption{}
\end{subfigure}
      \caption{Analysis of the principal pole $z_0(\lambda) \in \mathbf{P}_1^+$, as $(\sigma_1, \sigma_2)$ varies according to \eqref{eq:sigma-paramters}, with $\lambda \in [2, 100]$.  The dashed lines represent the approximating functions determined in Proposition \ref{prop:Rez_Imz}, while the solid lines illustrate the numerical computations of the poles for varying $\lambda$, performed using Mathematica. (a) Plot of $\Re z_0(\lambda)$.  (b) Plot of $\Im z_0(\lambda)$. (c) Plot of $1/(2\Re z_0 (\lambda))$, the width of the vertical stripes.
      }
\label{fig:biasing}
\end{figure*}

\subsection{Relationship with psychophysical experiments}
\label{sec:experiments}
Localized information is employed for example in \cite{tamekue2023mathematical, tamekue2024reproducibility, nicks2021understanding} to formally describe perceptual phenomena such as the Mackay effect and the Billock-Tsou illusion, using neural field models.

\subsubsection{MacKay effect}
The visual phenomenon known as the MacKay effect (\cite{mackay1957some}) occurs when the perception of radial and circular patterns is mixed. Specifically, we concentrate on the impact of observing a radial pattern, causing circles to appear superimposed on it. This psychophysical experience is illustrated in Figure \ref{fig:MK}, image (a).

\begin{figure*}[tbh]
\centering
\begin{subfigure}[b]{0.24\textwidth}
         \centering
         \includegraphics[width=\textwidth]{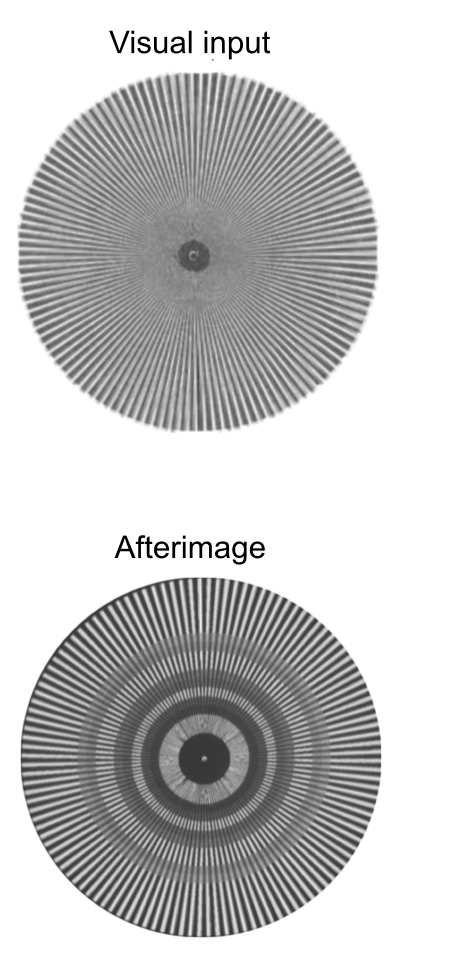} 
         \caption{}
\end{subfigure}
\begin{subfigure}[b]{0.24\textwidth}
         \centering
         \includegraphics[width=\textwidth]{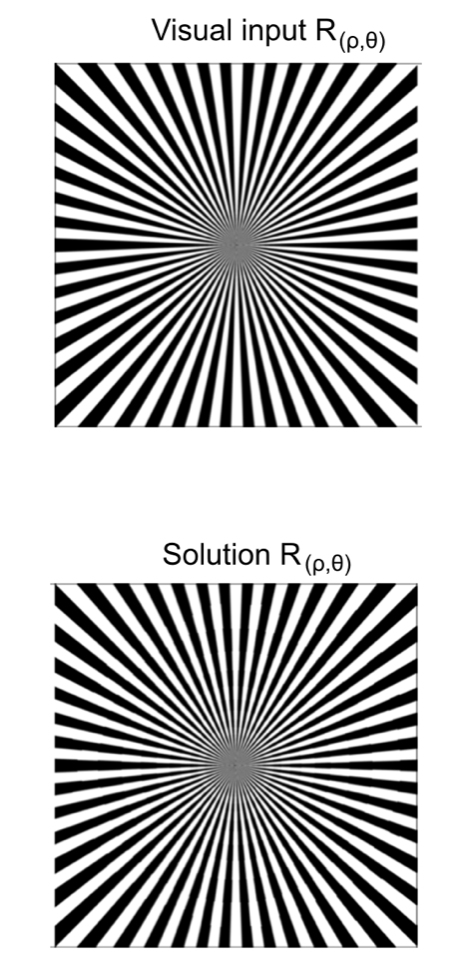} 
         \caption{}
\end{subfigure}\hspace{0.2cm}
\begin{subfigure}[b]{0.48\textwidth}
         \centering
         \includegraphics[width=\textwidth]{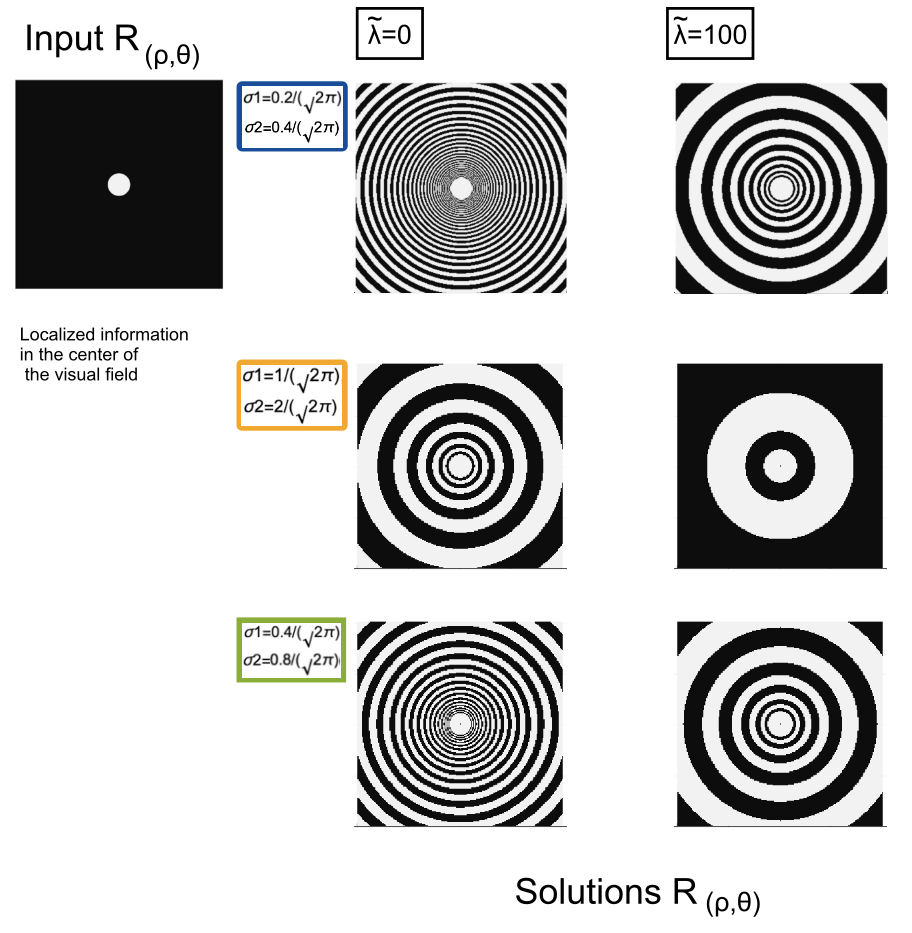} 
         \caption{}
\end{subfigure}
      \caption{Relationship with MacKay effect. (a) Depiction of MacKay effect (\cite{mackay1957moving}): the visual stimulus consists of radial patterns and it is displayed on the top. Looking at the center of the stimulus, leads to an afterimage of concentric rings superimposed on the background, as illustrated on the bottom (representation by \cite{leviant1996does}). (b) MacKay effect reproduced in \cite{tamekue2023cortical}: top image represents the binarized  input \eqref{eq:input-MK} (visual stimulus) in retinal coordinates. The bottom image shows the binary solution (the perceived image) to the corresponding linear mean field equation \eqref{eq:mean_field} in retinal coordinates. (c) Display in retinal coordinates of different binarized solutions \eqref{eq:flicker-sol} (for a fixed $\bar t$) corresponding to the binarized input \eqref{eq:input-flicker} describing localized information in the center of the visual field. Parameters $(\sigma_1, \sigma_2)$ are varied according to \eqref{eq:sigma-paramters} and $\lambda \in \{0, 100\}$. The variation in the size of the patterns translates to changes in the width of the perceived circular patterns in image (b) at the bottom.}
\label{fig:MK}
\end{figure*}

In \cite{tamekue2023mathematical}, the authors show that this experience can be reproduced using a neural field model governed by \eqref{eq:mean_field} with a linear firing rate function $f(s)=s$ and static visual input defined as:
\begin{equation}
\label{eq:input-MK}
    I(x) \coloneqq \cos(2\pi\zeta x_2) + H(-x_1),
\end{equation}
where $x =(x_1, x_2) \in \mathbb{R}^2$, $\zeta \in \mathbb{R}^+$. We recall that $H(-x_1)$ models localized information in the center of the visual field. 

\begin{remark} 
The authors in \cite{tamekue2023mathematical, tamekue2023cortical} show the existence of a globally attractive stationary stable state for \eqref{eq:mean_field}, under the assumptions defined at the beginning of Section \ref{sec:general} with a time-indpendent input. Then, combining the results from \cite{tamekue2023cortical, tamekue2023mathematical} with Corollary \ref{cor:0-poles}, the solutions of linear \eqref{eq:mean_field} with input \eqref{eq:input-MK} converges toward:
\begin{equation}
\begin{aligned}
u^\star&(x)\\
&=  \tilde I(x)+ V_0(x) + \mathcal{O}\left(\frac{e^{-2\pi |x_1| \Im{z_0}}}{x_1} \right), \\
\end{aligned}
\end{equation}
where $\tilde I(x) = (1+C_\zeta)\cos(2\pi\zeta x_2) + H(-x_1)$ with $C_\zeta = \hat{\omega(\zeta)}/(1-\hat{\omega}(\zeta))$ accordingly to \cite[Prop.1]{tamekue2023cortical}.
\end{remark}

The localized information in the center of the visual field, modeled by the Heaviside function, causes the emergence of vertical stripes in the complementary peripheral part. Section \ref{sec:biasing} investigates how the size of these illusory patterns depends on the parameters of the equation, specifically on the principal pole of the kernel $K_\ell$, for the appropriate $\ell$. By changing $(\sigma_1, \sigma_2)$ according to \eqref{eq:sigma-paramters}, the size of the perceived patterns modifies. 

MacKay's study (\cite{mackay1957moving}) notes that the effect persists under strobe lights. If the redundant information flickers, i.e., we substitute $H(-x_1)$ with $H(-x_1)\cos(\lambda t)$, the size of the patterns adjusts accordingly to \eqref{eq:flicker-sol}. This is illustrated in Figure \ref{fig:MK}, image (c), where we show the localized information in the center of the visual field and how the solution of the corresponding neural field equation changes with varying parameters $(\sigma_1, \sigma_2)$ and $\lambda$, in accordance with results presented in Section \ref{sec:biasing}.

\subsubsection{Billock-Tsou illusion}

Billock and Tsou's work on visual illusions (\cite{billock2007neural}) explores the interaction of flickering visual stimuli, such as a series of light flashes on a monitor, with geometric visual patterns. In particular, we focus on radial patterns introduced in the center of the visual field surrounded by peripheral flickering, raising the perception of concentric circles in the flickering peripheral area. An illustration of this effect can be found in Figure \ref{fig:BT-exp}. 

\begin{figure}[tbh]
    \centering
\includegraphics[width=0.45\textwidth]{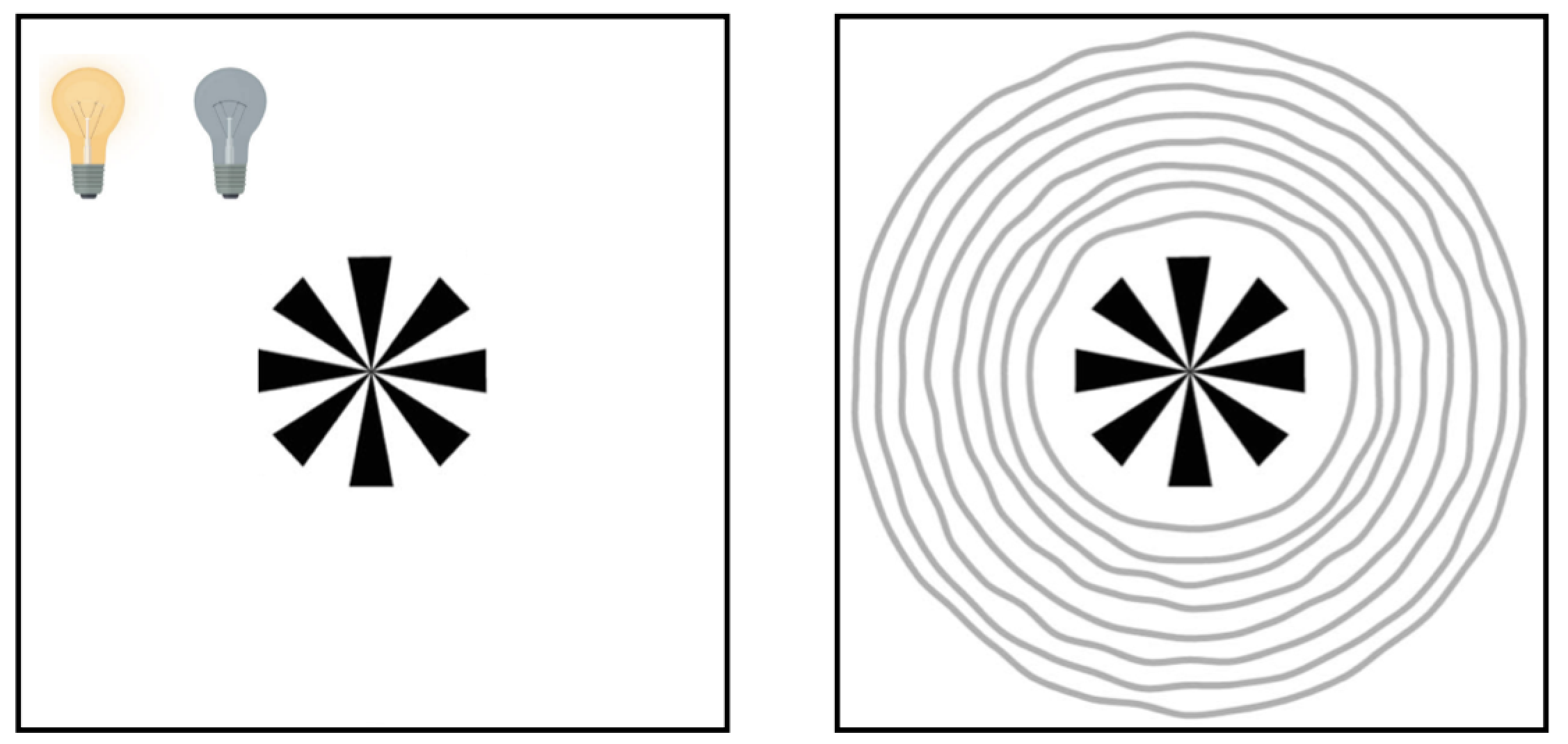}
\caption{Billock-Tsou's psychophysical experiment in the visual field. The visual stimulus is depicted in black: fan shapes pattern in the center of the image. The peripheral area to the stimulus is subject to flickering. Concentric circles represented in grey are the illusory effect within the flickering area. Image adapted from \cite{billock2007neural}.}
\label{fig:BT-exp}
\end{figure}

Previous works (\cite{nicks2021understanding, tamekue2023cortical, tamekue2024reproducibility}) have addressed the reproducibility of this phenomenon using mean field equations with a static input (i.e., no flicker in the complementary area to the geometric visual stimulus). In particular, \cite{tamekue2023cortical, tamekue2024reproducibility} show that reproducing the phenomenon is closely linked to the nonlinearity of the mean field equation \eqref{eq:mean_field}.

\begin{figure*}[tbh]
    \centering
\begin{subfigure}{\textwidth}
         \centering
         \includegraphics[width=0.9\textwidth]{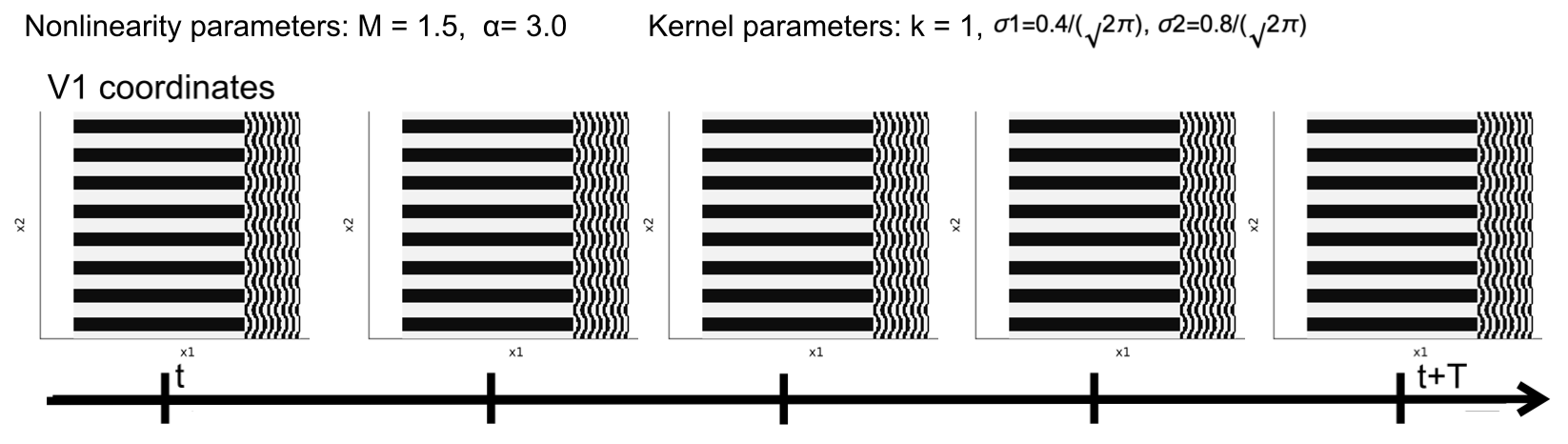}
         \caption{}
\end{subfigure}
\begin{subfigure}{\textwidth}
         \centering
         \includegraphics[width=0.9\textwidth]{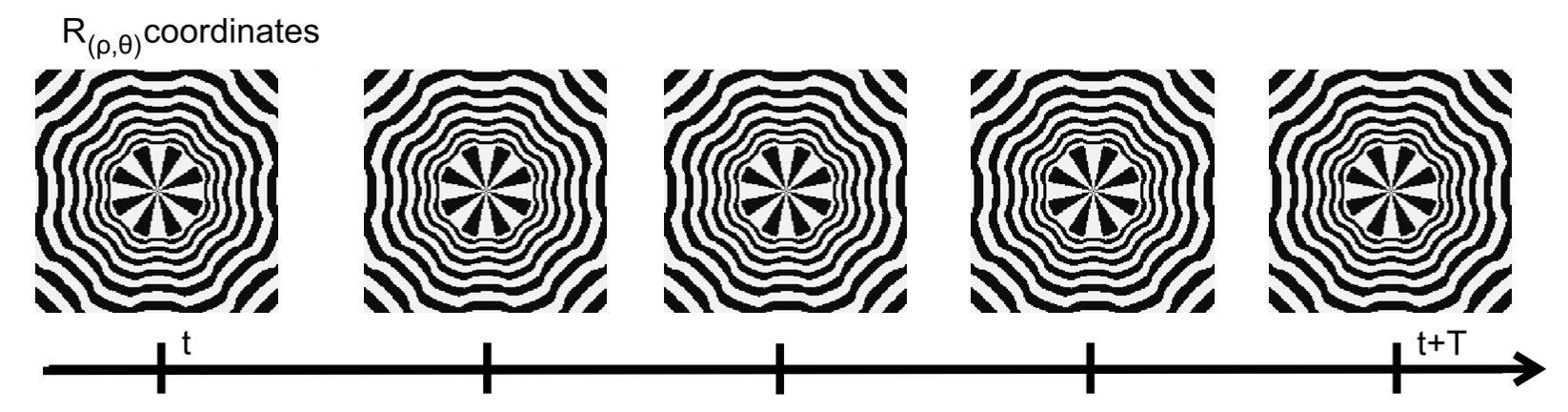} 
         \caption{}
\end{subfigure}\captionsetup{justification=justified, singlelinecheck=false, width=\linewidth}
\caption{Binary illustration of the periodic attractive state $ u^\star $ of \eqref{eq:mean_field} with modified input \eqref{eq:input-BT}, adjusted to accommodate a larger visual stimulus at the center of the visual field: $I(x, t)= H(5-x_1)\cos(2\pi \zeta x_2)+ H(x_1-5)\cos(\lambda t) $. We use the firing rate function $ f_{1.5, 3 }(s) = \max\{-1.5, \min\{3, s\}\} $, along with kernel parameters for \eqref{eq:DoG_kernel} chosen to be $k = 1, (\sigma_1, \sigma_2) = (0.4/\sqrt{2}\pi, 0.8/\sqrt{2}\pi) $ and flicker frequency $\lambda = 60 $. The evolution of the solution is depicted across frames over the time interval $ [t, t=T] $, where $ T = 2\pi/\lambda $. (a) Display of $u^\star$ in cortical coordinates.  (b) Display of $u^\star$ in retinal coordinates.
 }
\label{fig:BT}
\end{figure*}

\begin{remark}
Let us consider \eqref{eq:mean_field} with linear firing function $ f(s) = s $ and the following input:
\begin{equation}
\label{eq:input-BT}
I(x, t) = \underbrace{H(-x_1)\cos(2\pi \zeta x_2)}_{\text{static input}} + \underbrace{H(x_1)\cos(\lambda t)}_{\text{flickering input}}
\end{equation}
where $ x = (x_1, x_2) \in \mathbb{R}^2 $, $ \zeta, \lambda \in (0, \infty) $. The static input corresponds to the one used in \cite{tamekue2023cortical, tamekue2024reproducibility}. According to \cite[Prop. 5]{tamekue2023cortical} and \cite[Thm. 4.1]{tamekue2024reproducibility}, the solution with respect to the static input converges to a combination of horizontal and vertical stripe patterns in the peripheral area, instead of exclusively vertical stripes.
On the other hand, the solution corresponding to the flickering input converges to \eqref{eq:flicker-sol}, leaving the peripheral visual field unchanged. 
This aligns with the irreproducibility of the effect in the linear case, even if adding the flickering inputs.
\end{remark}

Numerical simulations performed in Julia allow us to solve the nonlinear mean field equation \eqref{eq:mean_field} with input \eqref{eq:input-BT}
and firing rate function $f_{M, \alpha }(s) = \max\{-M, \min\{1,\alpha s\}\} $, where $ M, \alpha > 0 $. Here, the temporal oscillation $ \lambda = 60 $ is chosen to align with frequencies known from psychophysical studies (e.g., \cite{shenyan2024visual, bartossek2021altered}) that induce significant illusory effects. Results are depicted in Figure \ref{fig:BT}, showing the reproducibility of the phenomenon with nonlinear parameters $M=1.5 $ and $ \alpha=3 $, and kernel parameters $k = 1, (\sigma_1, \sigma_2) = (0.4/\sqrt{2}\pi, 0.8/\sqrt{2}\pi)$. Circular patterns lightly move, suggesting slightly pulsation.

\section{Conclusions and future works}

In this paper, we study neural field dynamics in presence of time-periodic inputs, which model flickering visual stimuli.

We consider the neural field equation \eqref{eq:mean_field}, where parameters are chosen so that, in the absence of input ($I \equiv 0$), there exists a unique stable stationary solution. When the input $I$ is periodic in time, we show the exponential convergence toward a unique periodic solution $u^\star$ for \eqref{eq:mean_field}, which preserves the symmetries of the dynamic. This differs from the model presented in \cite{rule2011model} for flickering phosphenes, where spatially homogeneous inputs can lead to non-homogeneous solutions due to input impact on system stability.

We further show that when the $L^\infty$ norm of the input is small, the nonlinear dynamics \eqref{eq:mean_field} can be approximated by its linearization. 
So, in the linear case, we determine a decomposition for $u^\star$ in terms of the external input and a new kernel based on $\omega$. For $ x \in \R$, when $\omega$ is taken to be a difference of Gaussians (e.g., \eqref{eq:kernel}), we can explicitly write the analytic expression of these kernels, extending the technique introduced for the static case in \cite{tamekue2023mathematical, tamekue2023cortical} for the periodic case.

We then use this explicit representation to apply our theoretical results to visual processing. Geometric patterns commonly used in the study of visual illusions, such as concentric circles or radial shapes, are effectively one-dimensional when processed in the primary visual cortex (V1). Therefore, we formally study their interaction with localized flickering stimuli: 
in agreement with previous results (\cite{tamekue2023mathematical, tamekue2023cortical}), localized information generates the geometric percepts, while we show that flicker determines their perceived motion. Additionally, we emphasize the importance of the inhibitory component of the kernel along with flicker frequency in shaping the size of the perceived pattern.

Finally, we relate these findings to psychophysical experiments that are interpretable from a one-dimensional perspective, such as the MacKay effect (\cite{mackay1957moving}) and the Billock-Tsou illusion (\cite{billock2007neural}). Due to the explicit representation of the linear solution, we show flicker frequency and kernel parameters change the size of patterns in the MacKay effect; on the other hand, we recognize (in agreement with \cite{tamekue2024reproducibility}) 
that the Billock-Tsou illusion remains a result of nonlinear firing rate functions. Future studies will explore the influence of nonlinear shapes and parameter dependencies $(\sigma_1, \sigma_2)$ and $\lambda$ on the nonlinear Billock-Tsou effect.

\section*{Acknowledgments}
This work has been supported by the grant ANR-20-CE48-0003 of the Agence Nationale de la Recherche.

\begin{appendices}

\section{Technical results}
\label{app:proofs}

We begin with the following Gronwall's lemma, see for instance \cite[Proposition~2.1]{emmrich1999discrete} for a proof.
\begin{lemma}\label{lem::gronwall-lemma}
    Assume that $u\in C([0,T); \R)$, $T\in(0,\infty)$ satisfies for any $t\in[0,T]$ the integral inequality 
    \begin{equation}
        u(t)\le u(0)+\int_{0}^{t}g(s)u(s)ds+\int_{0}^{t}h(s)ds,
    \end{equation}
    for some $0\le g\in L^1(0,T)$ and $h\in L^1(0, T)$. Then, for any $t\in [0,T]$, $u$ satisfies the inequality
    \begin{equation}
        u(t)\le u(0)e^{G(t)}+\int_{0}^{t}h(s)e^{G(t)-G(s)}ds,
    \end{equation}
    where $G(t)=\int_{0}^{t}g(s)$.
\end{lemma}

\subsection{Computation of the 1D kernel}
In this appendix, we develop the proofs related to Section \ref{sec:linear}, building on complex analysis results.

The aim is to compute the inverse transform of $\hat{K_\ell}(\xi)$ when $\omega$ is defined as in \eqref{eq:DoG_kernel}, using the Residue theorem to evaluate the real integral of
\begin{equation}
\label{app:meromorphic}
\hat{K_\ell}(\xi)e^{2\pi i x \xi}= \frac{(i \ell -\hat{\omega}(\xi))e^{2\pi i x \xi}}{ 1 + i \ell  - \hat{\omega}(\xi)},\quad \xi \in \mathbb{R},  
\end{equation}
generalizing the strategy adopted in the static case in \cite{tamekue2023mathematical}.
We consider the set  
\begin{equation}
\label{app:poles}
P_\ell \coloneqq \{z \in \mathbb{C} \mid 1+ i\ell-\hat{\omega}(z) = 0\},
\end{equation}
which is the set of the poles of the meromorphic continuation of \eqref{app:meromorphic}. We recall that, as noted in \cite{tamekue2023mathematical} (see also \cite{berenstein2012complex}), given that $1+i\ell -\hat{\omega}(z)$ is a complex exponential polynomial, the assumption $\sigma_1^2/\sigma_2^2 \in \mathbb{Q}$ implies that $P_\ell$ is discrete and countable.

\begin{proposition}
\label{prop:app_poles}
Let $z\in P_\ell$, $\ell\in\Z$. Then,
\begin{equation}
    -z \in P_\ell
    \qquad\text{and}\qquad
    \bar{z}, -\bar{z} \in P_{-\ell},
\end{equation}
where $\bar{z}$ denotes the complex conjugate of $z$.
\end{proposition}

\begin{proof}
Since $ 1 + i \ell  - \hat{\omega}$ is an even function, $-z \in P_\ell$.
Then, we rewrite $ 1 + i \ell  - \hat{\omega}(z) = 0$ as:
\begin{equation}
\begin{aligned}
 i\ell + 1 &= e^{-2\pi^2\sigma_1^2 z^2} - ke^{-2\pi^2\sigma_2^2z^2};
\end{aligned}
\end{equation}
this is equivalent to:
\begin{equation}
\label{eq:decomposition_mD}
\begin{aligned}
1 &=& e^{-2\pi^2\sigma_1^2 (\Re{z}^2-\Im{z}^2)}\cos(4\pi^2\sigma_1^2\Re{z}\Im{z})+ \\&& -ke^{ - 2\pi^2\sigma_2^2 (\Re{z}^2-\Im{z}^2)}\cos(4\pi^2\sigma_2^2\Re{z}\Im{z}),\\
\ell &= &- e^{-2\pi^2\sigma_1^2 (\Re{z}^2-\Im{z}^2)}\sin(4\pi^2\sigma_1^2\Re{z}\Im{z})+\\ &&+  ke^{ - 2\pi^2\sigma_2^2 (\Re{z}^2-\Im{z}^2)}\sin(4\pi^2\sigma_2^2\Re{z}\Im{z}).
\end{aligned}
\end{equation}
It follows that if $z \in P_\ell$, then it satisfies \eqref{eq:decomposition_mD}, and the conjugate $\bar{z}$ satisfies \eqref{eq:decomposition_mD} for $P_{-\ell}$.
\end{proof}

We computed $K_\ell$ in Proposition \ref{prop:K_series} using the symmetries of the poles identified in Proposition \ref{prop:app_poles} to apply the Residue Theorem. First, we observed
\begin{equation}
\label{app:boh}
	K_\ell = 
	\begin{cases}
		2K_0^2 & \text{if } \ell=0,\\
		\sgn(\ell)K^1_\ell+K_\ell^2 & \text{if } \ell=0
	\end{cases}
\end{equation}
where the Fourier transforms of ${K}^1_\ell$ and ${K}^2_\ell$ are given by:
\begin{equation}
\label{eq:k1_2}
\begin{aligned}
\hat{K}^1_\ell(\xi) &= \frac{i |\ell|}{( 1 + i \ell  - \hat{\omega}(\xi))( 1 -i \ell  - \hat{\omega}(\xi))} \\
\hat{K}^2_\ell(\xi) &= \frac{\hat{\omega}(\xi)(\hat{\omega}(\xi)-1) +\ell^2}{( 1 + i \ell  - \hat{\omega}(\xi))( 1 -i \ell  - \hat{\omega}(\xi))}.
\end{aligned}
\end{equation}
The $K^\nu_\ell$ for $\nu = 1, 2$ are determined by switching from the exponential series \eqref{eq:limit_cycle} to its trigonometric form, exploiting the Convolution theorem. The expression \eqref{app:boh} is then obtained by returning to the exponential form. 

Notice that the set of poles of the meromorphic continuation of $K_\ell^\nu$, $\nu=1,2$, are exactly $P_\ell\cup P_{-\ell}$.
The residue of  $K_\ell^\nu(z) e^{2\pi i x z}$ in $z$ is computed as (e.g, see \cite{cartan1961theorie}):
\begin{multline}
\label{eq:res1}
\text{Res}(K_\ell^1(z) e^{2\pi i x z}, z) = \frac{i {|\ell|} e^{2\pi i x z}}{(2(1-\hat{\omega})\hat{\omega}')
(z)},\\
\text{Res}(K_\ell^2(z) e^{2\pi i x z}, z) = \frac{(\hat{\omega}(\hat{\omega} - 1)(z) +\ell^2) e^{2\pi i x z}}{(2(1-\hat{\omega})\hat{\omega}')(z)} .\\
\end{multline}
Observe that if $z \in P_{\pm \ell}$ then $(1-\hat{\omega})(z)= \mp i\ell$.
\begin{remark}
\label{rmk:cz}
Let us denote with $c_z = \hat{\omega}'(z)$, with $z \in \C$. 
Explicit writing the real and imaginary part, leads to  
\begin{equation}
\label{eq:ck}
\frac{1}{c_z} = \frac{\Re{c_z} - i \Im{c_z}}{\abs{c_z}^2}, \hspace{0.5cm} \frac{1}{c_{-\bar{z}}} = -\frac{\Re{c_z} + i \Im{c_z}}{\abs{c_z}^2}.
\end{equation}
\end{remark}

We now focus on the poles that lie in the positive quadrant: 
\begin{equation}
\label{eq:app_P+}
    \mathbf{P}^{+}_\ell \coloneqq \{ z \in P_\ell\cup P_{-\ell} 
    \mid \Re z > 0 , \Im z>0\};
\end{equation}
and since $\mathbf{P}_\ell^+$ is discrete, we pick the enumeration 
\begin{equation}
	\label{eq:app_poles++}
	\mathbf{P}_\ell^+=\{z_0,z_1,\ldots\}
\end{equation}
 such that $\Im z_i\le \Im z_{i+1}$ and, if $\Im z_i = \Im z_{i+1}$, it holds ${\Re z_i}\le {\Re z_{i+1}}$.

\begin{lemma}
\label{lem:K_series}
Assume $\mathbf{P}_\ell^{+}$ defined as in \eqref{eq:app_P+} with the enumeration provided in \eqref{eq:app_poles++}. The kernels defined in \eqref{eq:k1_2} can be recasted as:
\begin{equation}
    \label{eq:K1&K2_ell}
    \begin{aligned}
    K_{{\ell}}^1(x) &= -&2\pi i &\sum_{j=0}^\infty \frac{e^{-2\pi \abs{x} \Im{z_j}}}{\abs{c_{j}}} \sin(2\pi \abs{x} \Re{z_j}+ \phi), \\
        K_{{\ell}}^2(x) &=  &2\pi &\sum_{j = 0}^\infty \frac{e^{-2\pi \abs{x} \Im{z_j}}}{\abs{c_{j}}} \cos(2\pi \abs{x} \Re{z_j}+\phi),
        \end{aligned}
\end{equation}
where $\phi= \tan^{-1}\left( {\Re c_{j} }/{\Im c_{j} }\right)$.
\end{lemma}

\begin{proof}
To compute the inverse Fourier transforms of \eqref{eq:k1_2}, we use the Residue Theorem. More precisely, we will show that
\begin{equation}
\label{eq:inverseF}
\int_{\R} \hat{K}_\ell^\nu (\xi)e^{2\pi i x\xi}\,d\xi = \lim_{n\to +\infty} \int_{\Gamma_n} \hat{K}_\ell^\nu (z)e^{2\pi i xz}\,dz, 
\end{equation}
for a particular choice of contours $\Gamma_n\subset \mathbb{C}$.
By the Residue Theorem, we will then compute the integrals on the l.h.s.~by computing the residues of $\hat{K}_\ell^\nu$ at the poles contained in $\Gamma_n$.

Specifically, for $n \in \N$, we define $\Gamma_n := [-R_n, R_n] \cup C_n^+$, where $R_n = (M n + \varepsilon)^{1/\alpha}$, with $M, \varepsilon, \alpha$ as appropriate constants, $ [-R_n, R_n] \subseteq \R$, and $C_n^+ = \{z = R_n e^{i \theta} \in \C \mid \theta \in [0, \pi]\}$. Then, applying the same reasoning as \cite[Theorem B.1 ]{tamekue2023mathematical}, it is possible to show that:
 \begin{equation}
        \lim_{n \rightarrow \infty}\int_{C_n^+}\hat{K}_{\ell}^\nu(z) e^{2\pi i x z}dz = 0,
\end{equation}
from which it follows \eqref{eq:inverseF}.

On the other hand, the denominator of $\hat{K}_\ell^\nu , \nu = 1, 2$ is given by $(1+i\ell -\hat{\omega})(1-i\ell-\hat{\omega})$ and therefore the set of poles it is the same and it is given by $ P_\ell \cup P_{-\ell}$. We observe that the poles considered by the path are those of the set $\mathbf{P}_\ell^+$ and minus their complex conjugates. Then, denoting with $h(z)=\hat{K}_{\ell}^\nu(z)  e^{2\pi i x z}$, we have \eqref{eq:inverseF} equal to: 
\begin{equation}
\label{eq:res_thm_applicationR}
    2\pi i \sum_{j=0}^\infty \bigg(\Res(h(z), z_j)
    + \Res(h(z), -\bar{z_j}) \bigg). 
\end{equation}
Developing \eqref{eq:res_thm_applicationR} for $\nu = 1$, we have: 
\begin{equation}
\begin{aligned}
K_\ell^1(x)  =& 
- 2\pi i \sum_{j=0}^\infty  \frac{e^{-2\pi x \Im{z_j}}}{2\abs{c_{j}}^2}( e^{2\pi i \Re{z_j} x}(\Re c_{j}-i\Im c_{j})\\ & +e^{-2\pi x i \Re{z_j}}(\Re c_{j}+i\Im c_{j}) )\\
     &= - 2 \pi i \sum_{j=0}^\infty \frac{e^{-2\pi x \Im{z_j}}}{\abs{c_{j}}^2}(\Re c_{j}\cos(2\pi \Re{z_j} x)\\ &+ \Im c_{j}\sin(2\pi x \Re{z_j}))\\
     &= -2\pi i \sum_{j=0}^\infty \frac{e^{-2\pi x \Im{z_j}}}{\abs{c_{j}}} \sin(2\pi x \Re{z_j}+ \phi).
\end{aligned}
\end{equation}
Analogously for $\nu = 2$, we have:
\begin{equation}
\begin{aligned}
  K^2_\ell(x)   &= 2\pi \sum_{j=0}^\infty \frac{e^{-2\pi x \Im{z_j}}}{\abs{c_{j}}} \cos(2\pi x \Re{z_j}+\phi).
\end{aligned}
\end{equation}
The thesis follows, since $K_\ell^\nu$ for $\nu = 1, 2$ are an even function in $x$.
\end{proof}

Considering the enumeration defined in \eqref{eq:poles++}, the ``principal'' term of $K_\ell(x)$ computed in \eqref{eq:K1&K2_ell} is the one corresponding to $ z_0 \in \mathbf{P}_\ell^{+}$.

\begin{theorem}
\label{thm:K0}
Given $K_\ell^\nu(x)$, with $ \nu = 1, 2$, the following approximations hold: 
   \begin{multline}
        e^{2\pi \abs{x} \Im z_0} K_\ell^1(x) = -\frac{2 \pi  i}{\abs{c_{0}}} \sin(2\pi \abs{x} \Re z_0+\phi) + \mathcal{O}(1/x),\\
         e^{2\pi \abs{x} \Im z_0} K_\ell^2(x) = \frac{2 \pi }{\abs{c_{0}}}\cos(2\pi \abs{x} \Re z_0+ \phi)+ \mathcal{O}\textbf{\textbf{}}(1/x),
\end{multline}
where $z_0 \in \mathbf{P}_\ell^+$ and $\phi = \tan^{-1}\left( {\Re c_0}/{\Im c_0}\right)$.
\end{theorem}

\begin{proof}
We  re-write \eqref{eq:K1&K2_ell} as: 
 \begin{equation}
 \begin{aligned}
        K_\ell^1(x) =& -&2\pi i \frac{e^{-2\pi x \Im z_0}}{\abs{c_{0}}} \sin(2\pi \abs{x}\Re z_0+ \phi) + R^1(x)\\
         K_\ell^2(x) =& &2\pi  \frac{e^{-2\pi x \Im z_0}}{\abs{c_{0}}} \cos(2\pi \abs{x}\Re z_0+ \phi) + R^2(x)
\end{aligned}
\end{equation}
where
 \begin{equation}
 \begin{aligned}
    R^1(x) &= -2 \pi i &\sum_{j = 1}^{+\infty} \frac{e^{-2\pi x \Im{z_j}}}{\abs{c_{j}}} \sin(2\pi \abs{x} \Re{z_j}+ \phi_j)\\
    R^2(x) &= 2 \pi  &\sum_{j = 1}^{+\infty} \frac{e^{-2\pi x \Im{z_j}}}{\abs{c_{j}}} \cos(2\pi \abs{x} \Re{z_j}+ \phi_j).
\end{aligned}
\end{equation}

    \textbf{Step 1:} We start by showing that $\abs{c_{j}}$ is bounded from below by a constant and this holds for every $j \in \N$. 

    Suppose we are in the region of the complex plane where the poles are characterized by the real part smaller than the imaginary part; i.e. $\Re z_j \leq \Im z_j $, then we have: 
\begin{equation}
\begin{split}
    \abs{c_{j}} & = 2 \abs{z_j}\abs{2\pi^2\sigma_1^2e^{-2\pi^2\sigma_1^2 z_j^2}-2\pi^2\sigma_2^2e^{-2\pi^2\sigma_2^2 z_j^2} } \\
 & \geq 4 \pi^2 \sigma_1^2 |z_0| e^{-2\pi^2\sigma_1^2 (\Re z_{j}^2-\Im z_{j}^2)}\\
    & \geq 4 \pi^2 \sigma_1^2 |z_0|\coloneqq \tilde{C}.  
\end{split}    
\end{equation}

On the other hand, suppose the poles lie in the region of the complex plane, where the imaginary part is larger than the real part, i.e. $\Re z_{j} > \Im z_{j}$. Let us show that  
\begin{equation}
    \label{eq:claim}
    \limsup\limits_{j\rightarrow \infty}(\Re z_{j}^2-\Im z_{j}^2)= M.
\end{equation}
Indeed, this will imply that $|c_j|\ge 2|z_0|e^{-M}$, which together with the previous point will yield that $|c_j|$ is uniformly bounded from below.

To this purpose, we start by observing that
\begin{multline}
    \label{eq:bah}
  1=\Re \hat\omega(z_{j_n}) 
  \le e^{-2\pi^2\sigma_1^2(\Re{z_{j_n}}^2-\Im{z_{j_n}}^2)}\\
  + ke^{ - 2\pi^2\sigma_2^2(\Re{z_{j_n}}^2-\Im{z_{j_n}}^2)}
\end{multline}
Assuming, by contradiction, that there exists a subsequence $(j_n)_{n\in \N}$ such that \[\limsup\limits_{n\rightarrow \infty}(\Re z_{j_n}^2-\Im z_{j_n}^2)= +\infty,\] the r.h.s.\ of \eqref{eq:bah} tends to $0$ as $k$ tends to $+\infty$. This is a contradiction, which proves \eqref{eq:claim}.

\textbf{Step 2:} Estimate of $R^i(x), i = 1, 2$.
Since the function $\tau\mapsto e^{-2\pi x \tau }$ is decreasing, we have
\begin{equation}
\begin{aligned}
 |R^i(x)| &\leq \frac{2 \pi }{C} \sum_{\substack{j = 1 \\
    z_j \in P^{++}_0}}^{+\infty} e^{-2\pi x \Im z_{j}} \\
    & \leq \frac{2 \pi }{C} \int_{\Im{z_0}}^{+\infty}e^{-2\pi x \tau} d\tau  \\ &=\frac{1}{C} \frac{e^{-2\pi x \Im{z_0}} }{x}.    
\end{aligned}
\end{equation}

The result follows since $K_\ell$ is an even function in $x$.
\end{proof}

\end{appendices}

\bibliography{biblio}% common bib file

\end{document}